\documentclass[journal,10pt]{IEEEtran}
\usepackage{amsmath,amsthm,amssymb}
\usepackage{tabularx}
\usepackage{tikz}
\usepackage{newalg}
\usepackage{hhline}
\usepackage{subcaption}

\DeclareMathOperator{\F}{\mathbb F}

\theoremstyle{plain}
\newtheorem{lemma}{Lemma}
\newtheorem{theorem}{Theorem}

\newtheorem{example}{Example}

\newcommand{\E}[2]{\mathbf E_{{#1}}\left[{#2}\right]}

\newcommand{\mF}{\mathcal F}
\newcommand{\set}[1]{\left\{{#1}\right\}}

\newcommand{\floor}[1]{\left\lfloor{#1}\right\rfloor}
\newcommand{\abs}[1]{\left|{#1}\right|}

\DeclareMathOperator{\sgn}{sgn}

\begin{document}

%+Title
\title{Fast Block Sequential Decoding of Polar Codes}
\author{Grigorii Trofimiuk, ~\IEEEmembership{Student Member,~IEEE,} Nikolai Iakuba, ~\IEEEmembership{Student Member,~IEEE,} Stanislav Rets,  Kirill Ivanov, ~\IEEEmembership{Student Member,~IEEE,} Peter Trifonov, ~\IEEEmembership{Member,~IEEE}\thanks{
Copyright (c) 2015 IEEE. Personal use of this material is permitted. However, permission to use this material for any other purposes must be obtained from the IEEE by sending a request to pubs-permissions@ieee.org.

This work was supported by Government of Russian Federation (grant 08-08).

This work was partially presented at the International Workshop on Wireless Communication Systems'2015.

G. Trofimiuk, N. Iakuba and P. Trifonov are with ITMO University, Saint-Petersburg, Russia. K. Ivanov is with EPFL, Switzerland. 
 E-mail: \{gtrofimiuk,nyakuba,pvtrifonov\}@itmo.ru, kirill.ivanov@epfl.ch}}

\date{\today}
\maketitle

%-Title

\markboth{IEEE Transactions on Vehicular Technology,~Vol.~XX, No.~XX, XXX~2020}
{Trofimiuk et al: Fast Sequential Decoding of Polar Codes}

%+Abstract
\begin{abstract}
A reduced complexity sequential decoding algorithm for polar (sub)codes is described. The proposed approach relies on a decomposition of the polar (sub)code being decoded into  a number of outer codes, and on-demand construction of codewords of these codes in the descending order of their probability. Construction of such codewords is implemented by fast decoding algorithms, which are available for  many codes arising in the decomposition of polar codes. Further complexity reduction is achieved by taking hard decisions of the intermediate LLRs, and avoiding decoding of some outer codes.
 Data structures for sequential decoding of polar codes are described.

 The proposed algorithm can be also used for decoding of polar codes with CRC and short extended BCH\ codes. It has lower average decoding complexity compared with the existing decoding algorithms for the corresponding codes. \end{abstract}
%-Abstract
\IEEEpeerreviewmaketitle 
\sloppy
\section{Introduction}
%Polar codes is a family of capacity-achieving codes with  low-complexity construction, encoding and decoding algorithms \cite{arikan2009channel}.
%However, their finite-length performance is quite poor. Improved code constructions, like polar codes with CRC \cite{tal2015list} and polar subcodes \cite{trifonov2016polar,trifonov2017randomized}, were shown to outperform state-of-the-art LDPC and turbo-codes.  However, list decoding techniques are needed to perform near-maximum likelihood decoding of such codes \cite{tal2015list}.  

Polar codes are first capacity-achieving codes with low-complexity construction, encoding and decoding methods \cite{arikan2009channel}. Near maximum likelihood (ML) decoding can be performed with the Tal-Vardy successive cancellation list (SCL) decoding \cite{tal2015list}. However, finite-length performance of polar codes is quite poor, motivating thus development of improved constructions. Short polar subcodes \cite{trifonov2016polar,trifonov2017randomized} and CRC-aided polar codes \cite{tal2015list} were shown to outperform state-of-the-art LDPC and turbo codes under list decoding with small list size.

% Polar codes are first capacity-achieving codes with low-complexity construction, encoding and decoding methods \cite{arikan2009channel}. Near maximum likelihood decoding can be performed with Tal-Vardy successive cancellation list algorithm (SCL) \cite{tal2015list}. However, finite-length performance of polar codes is quite poor, which gives raise to improved constructions, e.g. polar subcodes \cite{trifonov2016polar,trifonov2017randomized} and CRC-aided polar codes \cite{tal2015list}, that were shown to outperform state-of-the-art LDPC and turbo codes under SCL algorithm with small list size.

The complexity of SCL decoding algorithm can be reduced by employing block decoding, i.e. joint processing of subsequent blocks of information symbols \cite{alamdaryazdi2011simplified,sarkis2014fast,sarkis2016fast,hashemi2017fast,hanif2017fast,ardakani2019fast}, or symbol-based decoding techniques \cite{xiong2016symboldecision}. The complexity of this method can be further reduced by constructing unrolled decoders with  simplified flow control logic \cite{giard2015gbits}.

Another approach is to utilize stack decoding  \cite{niu2012stack} or its improved version known as the sequential decoding algorithm (SDA) \cite{miloslavskaya2014sequential,trifonov2018score}. This algorithm avoids  construction of many useless low-probability paths in the code tree. For sufficiently high SNR, its complexity approaches that of the successive cancellation (SC) decoding algorithm with the performance close to that of the SCL method. Varied improvements of stack decoding  were also proposed in \cite{zhou2016successive, aurora2018lowcomplexity, song2019efficient,zhou2020efficient}.

An alternative way to implement decoding of polar codes is based on sphere decoding \cite{husmann2018reduced,zhou2019improved,guo2015efficient,niu2014low}. However, the complexity of this method grows quickly with code length, so that the results for this method have been reported only for very short codes.

The idea of joint processing of some blocks of information symbols was suggested in \cite{alamdaryazdi2011simplified} in the context SC decoding and generalized in \cite{sarkis2016fast} for the case of SCL decoding. In this paper we extend this approach to the case of sequential decoding. We show that the proposed method, referred to as block sequential decoding algorithm (BSDA), can provide
% Many of the implementation tricks considered in \cite{sarkis2016fast} in %the context of list decoding are applicable in the case of sequential decoding %as well. We show that by combining them (as well as some new ones) with %the principle of sequential decoding, one can obtain  
the performance close to that of the SC list decoder with large list size with complexity approaching (at high SNR) that of the unrolled SC decoder. The proposed approach can be used both for polar subcodes and polar codes with CRC.
Furthermore, we show that, by exploiting the representation of a linear code via a system of dynamic freezing constraints, the proposed approach can be used for decoding of other error-correcting codes. In particular, we show that for a $(128,64,22)$ extended primitive narrow-sense BCH (eBCH) code the proposed algorithm provides better performance and lower complexity compared with a recent trellis-based sequential-type algorithm \cite{chen2015algorithms}.

The paper is organized as follows. The background on polar codes and the decoding algorithms is presented in Section \ref{sBackground}. The BSDA is introduced in Section \ref{sBlock}. The algorithms for decoding of some outer codes are presented in Section 
\ref{sOuterDecoding}. The construction of polar subcodes and its processing in the proposed algorithm is discussed in Section 
\ref{sPolSub}. Complexity analysis is provided in Section \ref{sComplexity}.
Simulation results are presented in Section \ref{sNumRes}. The implementation details of low-level data structures are described in the Appendix.

\section{Background}
\label{sBackground}
\subsection{Polar codes and Plotkin decomposition}
For a positive integer $n$, denote by $[n]$ the set of $n$ integers $\{0,1,\dots\,n-1\}$.
An $(n=2^m,k)$ polar code over $\F_2$ is a linear block code generated by $k$ rows of 
matrix\footnote{Polar codes are typically defined with the bit-reversal permutation matrix. However, it is convenient here to omit it, since this results in a simpler description of the proposed decoding algorithm. } $A_m=F^{\otimes m}$, where  $F=\begin{pmatrix}1&0\\1&1
\end{pmatrix}$, $\otimes m$ denotes 
$m$-times Kronecker product of the matrix with itself \cite{arikan2009channel}.    Hence, a codeword of a classical polar code is obtained as 
$c_0^{n-1}=u_0^{n-1}A_{m}$, where $u_s^t=(u_s,u_{s+1},\dots,u_t)$, $u_i=0, i\in \mathcal F,$  $\mathcal F\subset[n]$ is the 
set of $n-k$ frozen symbol indices, which will be referred to as \textit{frozen set}, and the remaining symbols are set to the data symbols being encoded.

An $(n = 2^m, k)$ polar code $\mathcal C$ with the frozen set $\mathcal F$ can be represented  as a code obtained via Plotkin concatenation of polar codes $\mathcal C_0$ and $\mathcal C_1$, i.e.
\begin{equation}
\label{m:PlotkinDec}
\mathcal C = \set{(u+v|v)\:|\:u \in \mathcal C_0, v \in \mathcal C_1},
\end{equation}
where $\mathcal C_i$ is given by the frozen set
$\mathcal F_i$, $\mathcal F_0 = \mathcal F \cap [n/2]$, $\mathcal F_1 = \set{j\:|\:j + n/2 \in \mathcal F}$.
 Such representation will be referred to as \textit{Plotkin decomposition} (PD) of the code $\mathcal C$ into the codes $\mathcal C_0$ and $\mathcal C_1$. This decomposition can be applied recursively. After a number of decomposition steps, it results in some simple codes such as repetition, SPC, etc, which admit low complexity decoding [5], [7].

%Многократным применением разложение полярного кода можно представить его %как обобщенный каскадный и применить многошаговое декодирование.
 
\subsection{Sequential decoding of polar codes}
Consider the decoding of $(n,k)$ polar code. Let $u_0^{n-1}$ be the  vector of input symbols of the polarizing transformation used by the transmitter. Given a received noisy vector $y_0^{n-1}$, the SDA constructs a number of partial candidate  vectors (paths) $v_0^{\phi-1}\in \F_2^\phi, \phi\leq n$, where $\phi-1$ is referred to as a path \textit{phase}, then evaluates how close their continuations $v_0^{n-1}$ may be to the received sequence, and eventually produces a single codeword, being a solution of the decoding problem.

The algorithm makes use of a double-ended priority queue (PQ). 
A PQ is a data structure, which stores tuples $(M,v_0^{\phi-1})$, where $M=M(v_0^{\phi-1},y_0^{n-1})$ is the score of the path $v_0^{\phi-1}$,  and provides efficient algorithms for the following operations \cite{Cormen2001introduction}:
\begin{itemize}
\item push a tuple into the PQ;
\item pop  a tuple $(M,v_0^{\phi-1})$ (or just $v_0^{\phi-1}$) with the highest $M$;
\item pop  a tuple $(M,v_0^{\phi-1})$ (or just $v_0^{\phi-1}$) with the lowest $M$;
\item remove  a  tuple from the PQ.
\end{itemize}
We assume here that the PQ may contain at most $D$ elements. 

We employ the multilevel bucket PQ implementation \cite{yakuba2015multilevel}, which is much more efficient compared to the  heap-based approach \cite{zhou2016successive} in the context of sequential  decoding.

 If the decoder returns to a phase $i$ more than $L$ times, all paths shorter than $i + 1$ are also removed. The parameter $L$ affects the performance of SDA in the same way as list size in the SCL decoder.

The stack decoding algorithm for polar codes operates as follows \cite{chen2013improved}:
\begin{enumerate}
\item Push into the PQ a zero-length vector with score $0$. Let $q_0^{n-1}=0$, where $q_\phi$ is the counter for the number of visits to phase $\phi$.
\item Extract from the PQ a path $v_0^{\phi-1}$ with the highest score. Let $q_\phi\gets q_\phi+1$.
\item If $\phi=n$, return codeword $v_0^{n-1}A_{m}$ and terminate the algorithm.
\item If  the number of valid continuations $v_0^{\phi}$ of a path $v_0^{\phi-1}$ exceeds the amount of free space in the PQ, remove from it the element with the smallest score.
\item Compute scores $M(v_0^\phi,y_0^{n-1})$ of valid continuations $v_0^\phi$ of the extracted path, and push them into the PQ.  Let $\phi \leftarrow \phi + 1$.
\item If $q_{\phi}\geq L$, remove from the PQ all paths $v_0^{j-1}, j\leq \phi$.
\item Go to step 2.
\end{enumerate}
In what follows, by iteration we mean one pass of the above algorithm over steps 2--7.

The parameter $D \leq  kL$ affects the amount of memory needed by the sequential decoder.  In general, $D$ can be much less than $kL$, however, setting $D$ too small may results in performance degradation.

A score function $M(v_0^\phi,y_0^{n-1})$ can be obtained 
as a generalization of the Fano metric, which was introduced  for sequential decoding of convolutional codes \cite{massey1972variable}. In the context of polar codes, its approximate version can be written as \cite{trifonov2018score} 
\begin{align}
M(v_0^{\phi-1},y_0^{n-1})=&\underbrace{\sum_{i=0}^{\phi-1}\tau(S_m^{(i)}(v_0^{i-1}|y_0^{n-1}),v_i)}_
{R(v_0^{\phi-1}|y_0^{n-1})} -\Psi(\phi),
\label{mPathScore}
\end{align}
where 
$
\Psi(\phi)=\E{Y_0^{n-1}}{R(u_0^{\phi-1}|Y_0^{n-1})}
$
is the bias function, which can be pre-computed offline, $Y_0^{n-1}$ are the random variables corresponding to the received vector,
$$\tau(S,v)=\begin{cases}
0,&\sgn(S)=(-1)^v\\
-|S|,&\text{otherwise,}
\end{cases}$$
is the penalty function, and   $S_m^{(i)}(v_0^{i-1},y_0^{n-1})$ are the modified log-likelihood ratios (LLRs) \cite{trifonov2018score,balatsoukasstimming2015llrbased}, which are given by
\begin{align}
\label{mMinSum1}
S_{\lambda}^{(2i)}(v_0^{2i-1},y_0^{2^\lambda-1})=&Q(a,b)=\sgn (a)\sgn (b)\min(|a|,|b|),\\
\label{mMinSum2}
S_{\lambda}^{(2i+1)}(v_0^{2i},y_0^{2^\lambda-1})=&P(v_{2i},a,b)=(-1)^{v_{2i}}a+b,
\end{align}
where $\lambda$ is a layer,  
$a=S_{\lambda-1}^{(i)}(v_{0,e}^{2i-1}\oplus v_{0,o}^{2i-1},y_{0,e}^{2^{\lambda}-1})$, and $b=S_{\lambda-1}^{(i)}(v_{0,o}^{2i-1},y_{0,o}^{2^\lambda-1})$.

The first term of \eqref{mPathScore} is the total penalty of path $v_0^{\phi-1}$ for its deviation from the one given by the hard decisions based on LLRs $S_m^{(i)}(v_0^{i-1}|y_0^{n-1})$. The second term is the expected value of the first term under the assumption that path $v_0^{\phi-1}$ is correct. Bias term allows one to properly compare the paths of different length and results in a huge reduction of the average number of iterations performed by the stack algorithm 
\cite{trifonov2018score} compared with the original implementation  \cite{chen2013improved}.

Similarly to the case of sequential decoding of convolutional codes, the described algorithm does not necessarily perform ML decoding even for $L=\infty$. This is due to the bias term in the path score, which may cause the correct path to be removed, if its score drops too sharply at some early phase $\phi$.

\section{Block sequential decoding}
\label{sBlock}
We propose to reduce the complexity of sequential decoding by joint processing of blocks of input symbols of the polarizing transformation. Similar approach was suggested in \cite{sarkis2016fast} in the context of list decoding. However, we show that in the case of sequential decoding this provides some additional benefits. Most importantly, one does not need to construct immediately $L$ most probable codewords for each block. Instead, these codewords can be constructed on-demand, and in many cases just one codeword is sufficient. 
 
\subsection{Recursive decomposition of polar codes}
\label{sDecomp}
Let us consider decoding of an $(n, k)$ polar code $\mathcal{C}$. 
We  propose
to recursively apply PD to the code until one obtains codes which admit efficient decoding.
This results in a code decomposition tree  similar to that introduced in \cite{alamdaryazdi2011simplified,sarkis2013increasing}.

 Each non-leaf node of this tree corresponds to a code $\hat{\mathcal C}$, and two its children correspond to codes $\hat{\mathcal C}_{0}$ and $\hat{\mathcal C}_{1}$ obtained from its PD. Codes corresponding to the leaves of this tree will be referred to as \textit{outer codes}. Let $\mathcal V$ be the number of leaves in the tree. We enumerate outer codes $\mathcal C_\psi$ with indices $\psi \in [\mathcal V]$ in the ascending order from the leftmost to the rightmost leaf of the code decomposition tree (see Figure \ref{fPD}). 

Essentially, list and sequential algorithms recursively decompose   $(n,k)$ polar code $\mathcal
C$, until
codes of length $1$ are obtained.  Each of these codes corresponds to symbols
$u_\phi, 0\leq \phi<n$, where $\phi$ is  the phase number.
We propose to stop this recursion at some layers, and arrange  symbols $u_\phi$ into blocks, which  correspond to $(n_\psi = 2^{m_\psi},k_\psi,d_\psi)$ codes $\mathcal C_\psi, \psi\in [\mathcal V]$, obtained via PD, where $n_\psi$ is length, $k_\psi$ is dimension, and $d_\psi$ is minimum distance of $\mathcal C_\psi$.  The $i$-th block starts at phase $\phi_\psi-n_\psi+1$ and ends at phase $\phi_\psi$, 
%so that $\phi_{i1}=\phi_i$, $\phi_{i0}=\phi_i-n_i/2$, where $n_i=2^{m_i}$ %for some $m_i\leq m$,and $\phi=n-1$.   
Symbols within the block are processed jointly. This processing reduces to list decoding
of outer codes $\mathcal C_\psi$. Construction of such a decomposition can be simplified by employing the techniques suggested in \cite{hashemi2019rateflexible}.
\begin{figure}[h]
\begin{center}
\scalebox{0.65}
{
\parbox{0.5\textwidth}
{
\begin{tikzpicture}
\tikzstyle{code} = [draw,fill=none,minimum size=1.2em]
\node[code,minimum width=20em] (v2) at (-3,0) {$\mathcal C(16,10)$};
\node[code,minimum width=10em] (v1) at (-5.1,-1.5) {$\hat{\mathcal C}_0(8,6)$};
\node[code,minimum width=10em] (v3) at (-0.7,-3) {$\mathcal C_2(8,4)$};
\node[code,minimum width=5em] (v4) at (-6.1,-3) {$\mathcal C_{0}(4,3)$};
\node[code,minimum width=5em] (v5) at (-3.95,-3) {$\mathcal C_{1}(4,3)$};
\draw  (v1) edge (v2);
\draw  (v2) edge (v3);
\draw  (v1) edge (v4);
\draw  (v1) edge (v5);

\node (v6) at (-7.2,-4) {};
\node (v7) at (1.5,-4) {};
\draw  (v6) edge[->] (v7);
\node at (-5.2,-4.25) {3};
\node at (-3.1,-4.25) {7};
\node at (1.1,-4.25) {15};
\node at (1.5,-4.25) {$\phi_i$};
\end{tikzpicture}}}
\end{center}
\caption{Plotkin decomposition tree for $(16,10)$ code}
\label{fPD}
\end{figure}
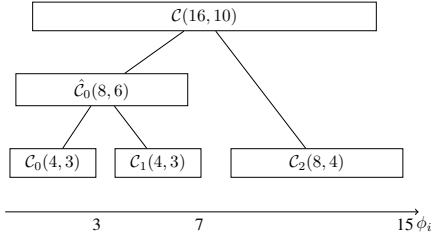

\begin{example}
\label{sExamplePD}
Consider the $(16,10)$ polar code $\mathcal C$ with frozen set $\mathcal F = \set{0,4,8,9,10,12}$. 

The PD\ tree of this code is shown in figure \ref{fPD}. One step of PD\ results in codes $\hat{\mathcal C_0}$ and $\mathcal C_2$. $(8,6)$ code $\hat{\mathcal C_0}$ is non-leaf in PD tree with the frozen set $\hat{\mathcal F_0} = \set{0,4}$. By applying  the PD to the code $\hat{\mathcal C_0}$, one obtains $(4,3)$ outer codes $\mathcal C_0 = \mathcal C_1$ with frozen sets $\mathcal F_0 = \mathcal F_1 = \set{0}$. These codes can be efficiently decoded (see sections \ref{ssHadamard} and \ref{ssSPC} for details). The $(8,4)$ code $\mathcal C_2$ with $\mathcal F_2 = \set{0,1,2,4}$ can also be efficiently decoded (see section \ref{ssHadamard}), hence we stop the recursion.
\end{example}

It remains to transform the path score function \eqref{mPathScore} into a form suitable for use with decoders of outer codes.
Let 
  $E(c_0^{n-1},S_0^{n-1})=\sum_{i=0}^{n-1}\tau(S_i,c_i)$
be the ellipsoidal weight\footnote{Conventionally, it is defined as a non-negative function. However, here we define it as a non-positive one to ensure consistency with the values which arise in the Tal-Vardy algorithm.} (also known as correlation discrepancy) of vector $c_0^{n-1}\in \F_2^n$ with respect to LLRs $S_0^{n-1}$   \cite{Valembois2004box,moorthy1997softdecision}. 
% \begin{lemma}
% \label{lWeightSum}
% For any  $c_0^{2n-1}\in\F_2^{2n-1}$
% $$E(c_0^{2n-1},S_0^{2n-1})=E(c_{0,e}^{2n-1}+c_{0,o}^{2n-1},\tilde S_0^{n-1})+E(c_{0,o}^{n-1},\overline S_0^{n-1}),$$
% where $\tilde S_i= Q(S_{2i},S_{2i+1})$, $\overline S_i=P(c_{2i},S_{2i},S_{2i+1})$.
% \end{lemma} 
% 

\begin{lemma}
\label{WeightSum}
For any $c_0^{2n-1}\in\F_2^{2n}$ one has 
$E(c_0^{2n-1},S_0^{2n-1})=E(c_{0}^{n-1}\oplus c_{n}^{2n-1},\tilde S_0^{n-1})+E(c_{n}^{2n-1},\overline S_0^{n-1}),$
where $\tilde S_i= Q(S_{i},S_{i+n})$, $\overline S_i=P(c_{i}\oplus c_{i+n},S_{i},S_{i+n})$.
\end{lemma} 
\begin{proof}
It can be seen  that  $E(c_0^{2n-1},S_0^{2n-1})=\sum_{i=0}^{n-1}E((c_i,c_{i+n}),(S_i,S_{i+n})).$ Hence, it is sufficient to prove the statement for   $n=1$.  Observe that $E(c_0^1,S_0^1)=\gamma=E((0,0),(S_0',S_1'))$, where $S_i'=(-1)^{c_i}S_i$. Hence, it is sufficient to consider the case of $c_i=0$.

For the case $Q(S_0',S_1')>0$ one has $\gamma=\tau(S_0',0)+\tau(S_1',0)=\tau(S_0'+S_1',0)$, while for $Q(S_0',S_1')<0$  one has $\gamma=\min(|S_0'|,|S_1'|)+\tau(S_0'+S_1',0)$. The latter equality follows by considering the cases of $S_0'+S_1'>0$ and $S_0'+S_1'\leq 0$. 
\end{proof}
\begin{theorem}\label{tLogLikelihood}
The ellipsoidal weight of the vector $u_0^{2^m-1}A_m$ with respect to the input LLRs $\mathbf S$ is equal to score of the path $u_0^{2^m-1}$ in the SC\ decoder, i.e. $E(u_0^{2^m-1}A_m,\mathbf S)=\sum_{i=0}^{2^m-1}\tau(S_m^{(i)}(u_0^{i-1},y_0^{2^m-1}),u_i),$
where $\mathbf S=(S_{0}^{(0)}(y_0),\dots,S_{0}^{(0)}(y_{2^m-1}))$.
\end{theorem}
\begin{proof}
For $m=0$, the statement is obvious. Let us assume that it is valid for some $m\geq 0$. Then, from Lemma \ref{WeightSum}, one obtains $E(u_0^{2^{m+1}-1}A_{m+1},\mathbf S)=E(u_0^{2^{m}-1}A_{m},\widetilde {\mathbf S})+E(u_{2^m}^{2^{m+1}-1}A_{m},\overline {\mathbf S})$, where $\widetilde {\mathbf S}_i=S_1^{(0)}(y_i,y_{i+2^m}), 0\leq i<2^m,$ and $\overline {\mathbf S}_i=S_1^{(1)}((u_0^{2^{m}-1}A_{m})_i,(y_i,y_{i+2^m}))$. Then the result follows from the inductive assumption.  
\end{proof}
Theorem \ref{tLogLikelihood}  implies that the  sum of those terms in \eqref{mPathScore}, which correspond to the same block in the PD tree, can be obtained by construction of a codeword of the corresponding outer code $\mathcal C_\psi$, and computing its ellipsoidal weight. In some cases this can be more efficient than performing iterations of SDA.

 Observe that the proposed approach can be also viewed as follows: we stop the calculation of LLRs $S_m^{(\phi_\psi)} = S_m^{(\phi_\psi)}(v_0^{\phi_\psi-1},y_0^{n-1})$, given by recursion \eqref{mMinSum1}--\eqref{mMinSum2}, at layer $m-m_\psi$ and decode $2^{m_\psi}$ LLRs $\mathbf S = 
S_{m-m_\psi}^{(r(\psi))}$ in code $\mathcal C_\psi$, where $r(\psi) =\lfloor \phi_\psi/2^{m_\psi} \rfloor $. Thus, Theorem \ref{tLogLikelihood} allows one to rewrite the  score \eqref{mPathScore} for the path constructed up to  block $\psi$ as 
\begin{equation}
M(v_0^{\phi_{\psi}},y_0^{n-1}) =\underbrace{R(v_0^{\phi_{\psi-1}}|y_0^{n-1}) + E(c,\mathbf S)}_{R(v_0^{\phi_{\psi}}|y_0^{n-1})}-\Psi(\phi_\psi),
\end{equation}
where $c = u_{\phi_{\psi-1}+1}^{\phi_\psi} A_{m_\psi}$is a codeword of $C_{\psi}$. 
\subsection{Outer codes}

The main idea of the proposed approach is to perform jointly the steps of the above described SDA, which correspond to the same block in the PD tree. Each combined step reduces to list decoding of the corresponding outer code  $\mathcal C_\psi, \psi\in \mathcal V$. 
%In this section we consider some additional details about introducing the outer codes to the SDA.

Observe that the decoder of outer code $\mathcal C_\psi$ may produce at most $2^{k_\psi}$ codewords. Since we consider the sequential algorithm, one does not need to obtain these codewords immediately. Instead, the codewords of outer codes can be constructed one-by-one, i.e. once the corresponding path is extracted from the PQ. In this case these codewords should be constructed in the descending order of their ellipsoidal weight.

Moreover, in most cases it is sufficient to obtain just two such codewords, which can be computed in a simpler way compared to the full list decoding the outer code. Observe that this simplification is not possible in the context of SCL decoding of the polar code.

We require that for each outer code $\mathcal C_\psi$ are available subroutines \textit{Preprocess}$(\mathcal C_\psi,\mathbf S,Z)$ and \textit{GetNextCodeword}$(\mathcal C_\psi,Z,\hat c)$.
The former performs some code-dependent preprocessing of LLR vector $\mathbf S$, and saves its results in a state variable $Z$. 
The latter uses $Z$ to construct the next most probable codeword  in the list, which is stored in the array given by pointer $\hat{c}$,  and returns tuple $[e,b]$, where $b$ is a boolean value, which is true iff more codewords can be obtained by the subsequent calls,  and $e=E(\hat c,\mathbf S)$. The structure $Z$ includes the following fields:
\begin{itemize}
\item $\mathbf S$ --- vector of LLRs. 
\item Any additional data needed for efficient recovery of codewords of  $\mathcal C_\psi$ for  given  $\mathbf S$. 
\end{itemize}

Note that the amount of codewords to be returned by \textit{GetNextCodeword}$(\mathcal C_\psi,Z,\hat c)$ is upper bounded by $\min(2^{k_\psi},L)$.
\subsection{The algorithm}

% \item One should stop construction of codewords of a particular outer code if $L_0$ codewords have already been obtained, or the ellipsoidal weight of the next codeword to be obtained exceeds some threshold. These rules ensure that the decoder does not spend too much time on incorrect paths in the code tree.

\begin{figure}
\small
\begin{subfigure}{0.5\textwidth}
\begin{algorithm}{BSDA}{S_0^{n-1}, L, D}
%\CALL{Cleanup}()\\
\CALL{Initialize}()\\
\begin{WHILE}{true}
(M,l) \leftarrow \CALL{PopMax}()\\
\begin{IF}{r(\psi_l -1) \text{is odd}}
\CALL{IterativelyUpdateC}(l,m-m_{\psi_l-1},r(\psi_{l-1}))
\end{IF}\\
\begin{IF}{\psi_l = \mathcal{V}}
\RETURN{\CALL{GetArrayPointerC\_R}(l,0,0)}     
\end{IF}\\
\begin{IF}{B_l = 1} 
\CALL{RemoveBadPaths}(D)\\
\CALL{BackwardPass}(l) \text{//path cloning}
\end{IF}\\
\CALL{IterativelyCalcS}(l,m-m_{\psi_l},r(\psi_l))\\
\CALL{ForwardPass}(l)\text{//extending the current path}\\
q_{\psi_l}\=q_{\psi_l}+1\\
\begin{IF}{q_{\psi_l} \ge L}
\begin{FOR}{\text{All paths $l'$ stored in the PQ}}
\begin{IF}{\psi_l\leq \psi_{l'}}
\CALL{KillPath}(l')\\
\text{Remove $l'$ from the PQ}
\end{IF}
\end{FOR}
\end{IF}
\end{WHILE}
\end{algorithm}
\caption{The algorithm}
\label{fBSDA}
\end{subfigure}
\vspace{2mm}
\begin{subfigure}{0.5\textwidth}
\begin{algorithm}{Initialize}{}
l \leftarrow \CALL{AssignInitialPath}()\\
\CALL{PushPath}(0,l)\\
q_0^{\mathcal V-1}\leftarrow 0, \psi_l \leftarrow 0, R_l \leftarrow 0,
B_l\leftarrow 0\\
%[q_0^{\mathcal V-1}, \psi_l, R_l,B_l] \leftarrow [0, 0,0,0]\\
s \leftarrow \CALL{GetArrayPointerS\_W}(l,0)\\
s[i] \leftarrow S_i,0\leq i<n
\end{algorithm}
\caption{Initialization of the algorithm}
\label{fInitBSDA}
\end{subfigure}
\vspace{2mm}
\begin{subfigure}{0.5\textwidth}
\begin{algorithm}{ForwardPass}{l}
\mathbf S \leftarrow \CALL{GetArrayPointerS\_W}(l,m-m_{\psi_l})\\
\hat{c} \leftarrow \CALL{GetArrayPointerC\_W}
(l,m-m_{\psi_{l}},r(\psi_l))\\
\CALL{Preprocess}(\mathcal C_{\psi_l},\mathbf S,Z_l)\\
[e,b]\leftarrow \CALL{GetNextCodeword}(\mathcal C_{\psi_{l}},Z_l,\hat{c})\\
B_l \leftarrow\ b; \tilde R_l \leftarrow R_l;R_l \leftarrow R_l-e\\ 
%[B_l,\tilde R_l, R_l] \leftarrow [b, R_l, R_l -e]\\
\CALL{Push}(R_{l}-\Psi(\phi_{\psi_l}),l)\\
\psi_l \leftarrow \psi_l + 1
\end{algorithm}
%\vspace{3.8cm}
\caption{Construction of the most probable codeword of outer code}
\label{fForwardPassBSDA}
\end{subfigure}
\begin{subfigure}{0.5\textwidth}
\begin{algorithm}{BackwardPass}{l}
l' \leftarrow \CALL{ClonePath}(l)\\
\hat{c}\leftarrow \CALL{GetArrayPointerC\_W}(l',m-m_{\psi_{l-1}},r(\psi_l -1))\\
[e,b]\leftarrow \CALL{GetNextCodeword}(\mathcal C_{\psi_l-1},Z_l,\hat{c})\\
%[B_{l'}, Z_{l'}, R_{l'},\tilde R_{l'},\psi_{l'}] \leftarrow [b, Z_l,R_l-e,\tilde R_l,\psi_{l}]\\
B_{l'}\=b;
Z_{l'}\=Z_l; 
R_{l'} \leftarrow \tilde R_l-e;\tilde R_{l'}\leftarrow\tilde R_l;\psi_{l'} \leftarrow \psi_{l}\\
\CALL{Push}(R_{l'}-\Psi(\phi_{\psi_l-1}),l')
\end{algorithm}
\caption{On-demand construction of codewords of outer codes}
\label{fBackwardPassBSDA}
\end{subfigure}
\caption{Block sequential decoding algorithm}
\end{figure}

\begin{table}
\caption{Variables used in BSDA}
\begin{tabular}{|>{$}p{0.05\textwidth}<{$}|p{0.37\textwidth}|}\hline
\text{Variable}&Description\\\hline
l&index of a path $v_0^{\phi_{\psi_l}}$\\
q_i&Number of invocations of the $i$-th outer decoder\\
\psi_l&The index of outer decoder to be invoked for the $l$-th path\\
\phi_i&The last phase of the $i$-th block\\
B_l& True if the $l$-th path should be cloned\\
m_j&$=\log_2 n_j$, where $n_j$ is the length of outer code $\mathcal C_j$\\
R_l&Accumulated penalty $R(v_0^{\phi_{\psi_l}}|y_0^{n-1})$ for the $l$-th path \\
\tilde R_l&$=R(v_0^{\phi_{\psi_l-1}}|y_0^{n-1})$ \\
Z_l&Saved state for the last outer decoder used for the $l$-th path \\
M&Score of a path\\
r(\psi)& $\floor{\phi_{\psi}/2^{m_{\psi}}}$
\\\hline
\end{tabular}
\label{tVars}
\end{table}

Figure \ref{fBSDA} illustrates the proposed block sequential decoding algorithm (BSDA).
Table \ref{tVars} presents the description of some of its internal variables. The input arguments for the algorithm are the LLRs $S_i=\log\frac{W(y_i|0)}{W(y_i|1)}$, where $y_i$ is the result of transmission of codeword symbol $c_i$ over a memoryless output-symmetric channel, maximal number of times $L$ the decoder is allowed to pass via any phase or block, and maximal total number $D$ of paths, which can be stored in the PQ.

Let us provide the brief description of the proposed algorithm. The algorithm makes use of the Tal-Vardy list decoder data structures \cite{tal2015list}. The implementation based on the original ones is described in \cite{trofimiuk2015block}. In this work we introduce some modifications, which are discussed in the Appendix. They avoid data copying and simplify the interface to outer decoders.

The algorithm starts from subroutine 
\textit{Initialize} (see figure \ref{fInitBSDA}), where the decoding data structures are initialized, input LLRs $S_i$ are loaded and the initial path is pushed into the PQ. 

The main loop of \textit{BSDA} starts from extraction of path $l$ with the best score $M$ from the PQ. After that, the path clone operation is done (if it is possible) in line 10 in the \textit{BackwardPass} function (see figure 
\ref{fBackwardPassBSDA}). Then, the most probable continuation of the path $l$ is constructed in the function \textit{ForwardPass} (see figure \ref{fForwardPassBSDA}).
After that, if $q_{\psi_l}\geq L$ then shorter paths are deleted from the PQ. Iterations are performed until a codeword is obtained and returned in line 7 of the \textit{BSDA} function.

Below we discuss the algorithm in more details. We denote by $\phi_\psi$ the index of the last input symbol corresponding to the $\psi$-th outer  $(n_\psi = 2^{m_\psi}, k_\psi)$ code $\mathcal C_\psi$, $\psi \in [\mathcal V]$.  The partial  sums of the input symbols $v_i$ of the polarizing transformation, which are needed for computing of $S_m^{(i)}(v_0^{i-1},y_0^{n-1})$, are updated in line 5, where $r(\psi) = \floor{\phi_{\psi}/2^{m_{\psi}}}$. Observe that in our algorithm we update the partial sums only for paths which were extracted from the PQ, while in case of list decoding this should be performed for each path in the list.   

The boolean variable $B_{l}$ is set to true iff at least one more codeword of code $\mathcal C_{\psi_l-1}$ can be returned by the corresponding outer decoder. In this case the decoder ensures in line 9 (\textit{RemoveBadPaths} procedure) that there are at most $D-2$  entries in the PQ (if not, the paths with lowest scores are killed), and  calls to \textit{BackwardPass} function. This function constructs the next most probable codeword of $\mathcal C_{\psi_l-1}$. This variable is set in the \textit{ForwardPass} and \textit{BackwardPass} functions. 

In line 11 the vector of LLRs $\mathbf S$ is computed. The decoder makes a call to the \textit{ForwardPass} algorithm, which constructs the most probable continuation of the $l$-th path, i.e. performs (near)\ maximum likelihood decoding of vector $\mathbf S$ in outer code. If the number of times $q_{\psi_l}$ the decoder has visited the $\psi_l$-th block exceeds $L$, then paths shorter than $\phi_{\psi_l}$ are removed in line 18.
 The first steps of \textit{ForwardPass} algorithm are to obtain writable pointers to the array $\mathbf S$ of log-likelihood ratios $S_{m-m_{\psi_l}}^{(r(\psi_l))}$, computed by \textit{IterativelyCalcS}, and to the array $\hat c$, which is used to store the most probable continuation of the $l$-th path. In line $3$ an appropriate pre-processing algorithm for $\mathcal C  _{\psi_l}$ is invoked (see Section \ref{sOuterDecoding} for details), and the most probable codeword is constructed in line 4.  Variable $e$ is assigned to the ellipsoidal weight of this codeword, while $b$ is set to $true$ iff less probable codewords can be obtained by \textit{GetNextCodeword} function. Finally, the value $R_l= R(v_0^{\phi_{\psi_l}}|y_0^{n-1})$, is updated according to Theorem \ref{tLogLikelihood}, and the path is pushed to the priority queue. The previous value of $R_l$ is saved in $\tilde R_l$, so that it can be used later to obtain the score of less probable continuations of this path.

%\begin{figure}
%\small
%\end{figure}

The \textit{BackwardPass} algorithm  is used to obtain less probable codewords of outer codes in the descending order of their ellipsoidal weight. At line 1 the path is cloned. A writable pointer to the destination array for storing the codeword is obtained in line 2, and an appropriate codeword of the outer code is stored in this array.

The details of low-level functions \textit{GetArrayPointer*} used in the proposed algorithm are discussed in the Appendix.

\begin{example}
Consider decoding of the $(16,10)$ polar code $\mathcal C$ from the Example \ref{sExamplePD} in AWGN channel at  $E_b/N_0=5$ dB. 

We  need the values  of bias function ${\Psi}(3)\approx -0.47, 
{\Psi}(7)\approx -0.52,{\Psi}(15)\approx -0.56$. 
 Let the input LLRs $S_0^{(0)}$ be equal to $(0.44,7.46,7.19,2.82,5.63,9.78,6.06,
 -0.12, -0.64,9.38$,
 $10.87,13.0,13.43,9.43,2.02,13.2).$ Let $l=0$ be the index of the initial path. At the first iteration, in line 11 the decoder computes the vector of LLRs $S_2^{(0)}$, which equals to $(-0.44,7.46,2.02,-0.12)$. \textit{ForwardPass} function obtains codeword $(1,0,0,1)\in\mathcal C_{0}$ with the ellipsoidal weight $e=0$. Hence, in line 6 of the \textit{ForwardPass} function a path with score $0.47$ is pushed to the PQ. 

This path is extracted from the PQ at the next iteration of \textit{BSDA}. \textit{BackwardPass} function obtains codeword $(0,0,0,0) \in \mathcal C_0$ with $e=-0.56$. The path is cloned (let the ID of the cloned path be $l'=1$), and an entry with score $-0.56 +0.47=-0.09$ is pushed to the PQ. 

The vector of LLRs $S_2^{(1)}$, given by $(6.08,16.89,9.2,-2.94)$, is obtained at line 11 for path $0$.  Hence, one obtains codeword $(0,0,0,0) \in \mathcal C_1$ with $e=-2.94$ by \textit{ForwardPass} function, and path $0$ is pushed to the PQ with score $-2.94+0.52=-2.42$.

At the next iteration of the decoder, path $l=1$ is extracted from the PQ. The vector of LLRs $S_2^{(1)}$, given by $(5.19,16.89,9.2,2.7)$, is obtained in line 11. The codeword $(0,0,0,0) \in \mathcal C_{1}$ is obtained with $e=0$, and path $1$ is pushed to the PQ with score $-0.09-0+0.52=0.43$. 

This path is extracted  from the PQ at the next iteration. The LLRs $S_1^{(1)}$ are equal to $(-0.2,16.84,$$18.05,$$15.82,$ $19.06,$$ 19.2,8.08,13.08)$. These values are preprocessed by the decoder for code $\mathcal C_2$, and the all-zero codeword with $e=-0.2$ is obtained in line 4 of the \textit{ForwardPass} function. Hence, path $1$ is pushed to the PQ with score $-0.09-0.2+0.56=0.27$. 

This path is extracted at the next iteration of the decoder, and, since all leaf nodes in the PD tree have been visited, the decoder terminates returning the all-zero codeword.
\end{example}    
The proposed algorithm can be tailored to implement decoding of polar codes with CRC. To do this, one should add CRC validation to line 7 of the BSDA, so that iterations are performed until either a correct codeword is found, or no more paths remain in the PQ.

The proposed algorithm is not guaranteed to provide the same performance as the original SDA. In some cases its performance may be better, since the decoders for outer codes may avoid some errors of the sequential decoder. However, in some cases performance degradation may occur, if it happens that for an incorrect path $v_0^{n-1}$ and some $ i$ $ \forall j>i:M(v_0^{\phi_j},y_0^{n-1})>M(u_0^{\phi_i},y_0^{n-1}), $ and $\exists \tau\in (\phi_{j-1},\phi_j]: \forall s\geq \phi_i \,\,\, M(v_0^{\tau},y_0^{n-1})<M(u_0^{s},y_0^{n-1}).$
That is, the proposed algorithm may miss the opportunity to switch to the correct path at an intermediate phase $\tau$ within some block, and proceed with exploration of an incorrect path. However, simulation results presented below show that the impact of this problem is negligible. 

%\subsection{Improvements}
%\subsubsection{Hard decisions}
\subsection{Hard decisions}
In many cases the hard decision vector corresponding to some intermediate LLR vector $\mathbf S$  is error free, i.e. it is a codeword of $\mathcal C_i$. In this case one should avoid invoking a relatively complex soft-decision decoding algorithm of outer code, i.e. \textit{Preprocess} and \textit{GetNextCodeword} functions in lines 3-4 of \textit{ForwardPass}, unless non-ML codewords of the corresponding outer code are needed.

Consider some decoding iteration and suppose that  path $l$ is extracted from the PQ. Let us construct the hard decision vector $\bar c$ of $\mathbf S$. If $\bar c \in \mathcal C_i$, then we can immediately set $e \leftarrow 0$, $b \leftarrow 1$ and push the path to the PQ. The LLR vector $\mathbf S$ is saved in the state variable $Z_l$, so that computationally expensive pre-processing can be done later.  
 
If the hard decision vector is a valid codeword of the corresponding outer code, then it is very likely that the less probable codewords will not be needed during the next iterations of the decoding. Hence, it is possible to skip construction of such codewords. However, occasionally such codewords may be needed, and some provision needs to be done in order to recover them later. It can be easily seen that the ellipsoidal weight of any such codeword cannot be more than $-d_{\psi_l-1}\min_i|Z_l.\mathbf S_i|$, where $d_{\psi_l-1}$ is the minimum distance of $\mathcal C_{\psi_l-1}$. We propose to use this value for computing an estimate of  $R_{l'}$ of  the less probable path $l'$. If this path is later selected by the decoder for further processing, the corresponding codeword should be actually constructed. 

% \subsubsection{Ellipsoidal weight threshold}
% In order to reduce the number of iterations performed by the decoder and the cost of operations with the priority queue, it is beneficial to set  $B_{l'}=0$ in the end of $BackwardPass$ algorithm,
% if $R_{l'}<R_l-\Delta$, where $\Delta$ is some parameter. This ensures that the decoder does not expand subtrees of the code tree corresponding to low-probability codewords. 

\section{Decoding of outer codes}
\label{sOuterDecoding}
As described in Section \ref{sDecomp},
PD is applied recursively until one obtains outer codes, which allow efficient ML decoding. Consider some $(n,k)$ outer code $\mathcal C$. We need to construct a decoder, which can find the codewords $c^{(i)}\in \mathcal C$ in the increasing order of their ellipsoidal weight $E(c^{(i)},S_0^{n-1})$, where $S_0^{n-1}$ is the vector of LLRs. In \cite{trofimiuk2015block} outer codes are decoded with tree-trellis Viterbi algorithm. However, in many cases it is possible to use much simpler algorithms. 

 In this section we describe the decoding algorithms for outer codes, which frequently arise in PD\ of polar codes. 
 Some of the techniques presented below resemble those suggested in \cite{sarkis2016fast}, but we also consider some well-known  outer codes, most importantly first-order Reed-Muller and extended Hamming codes. 

\subsection{Low rate codes}
Decoding of $(n, 0), (n,1)$ and $(n,2)$ codes is performed by exhaustive enumeration of their codewords $c^{(i)}$, computing the corresponding ellipsoidal weight $E(c^{(i)},S_0^{n-1})$ for each codeword, and sorting them in the ascending order of $E(c^{(i)},S_0^{n-1})$.

%Observe that for any $(4,3)$ code the same technique can be employed. If %the $(4,4)$ code decoded as described in Section \ref{ssRate-1} 
\subsection{First order Reed-Muller and related codes}
\label{ssHadamard}
The first order Reed-Muller code $RM(1,\mu)$  is obtained as  a polar code with the frozen set $\hat{\mathcal F}=[2^\mu]\setminus\left(\set{0}\cup\set{2^i|0\leq
i<\mu}\right)$. List decoding of such codes can be implemented using the fast Hadamard transform (FHT) with complexity $O(n\log n)$ \cite{green1966serial}. 
FHT computes correlations $T(c^{(i)},S_0^{n-1})=\sum_{j=0}^{n-1}(-1)^{c^{(i)}_j}S_j$ for $n$ codewords of the corresponding codes. The correlations for the remaining codewords are given by $T(c^{(i+n)},S_0^{n-1})=-T(c^{(i)},S_0^{n-1})$, and $c^{(i+n)}=c^{(i)}+\mathbf 1$, where $\mathbf 1$ is a vector of 1's.
The ellipsoidal weight of a codeword is related to its correlation by 
$E(c^{(i)},S_0^{n-1})=\frac{1}{2}\left(\sum_{j=0}^{n-1}\abs{S_j}-
T(c^{(i)},S_0^{n-1})\right).$

Observe that obtaining two most probable codewords, which are in most cases sufficient for the BSDA, requires finding just two highest values $T(c^{(i)},S_0^{n-1})$.

Another type of outer codes, commonly arising in the PD of polar codes, is a concatenation of a first order Reed-Muller code $RM(1,\mu-t)$ and a $(2^t,1,2^t)$ repetition code. Such codes may be also decoded using the FHT of order $2^{\mu-t}$.
 We propose also to use FHT-based decoder for the case of codes given by a union of at most 4 cosets of a first order Reed-Muller code $\mathcal R$, i.e.
$\mathcal C=\mathcal R\cup(\mathcal R+c'),$
and
$\mathcal C=\mathcal R\cup(\mathcal R+c')\cup (\mathcal R+c'')\cup (\mathcal R+c'+c''),$
where $c',c''\notin\mathcal R$. This turns out to be more efficient in practice than  performing additional steps of PD. 

\subsection{Single parity check code}
\label{ssSPC}
\begin{table}
\caption{Test error patterns for single parity check code}
\begin{tabular}{|c|p{0.41\textwidth}|}\hline
$\mathbf T^{(1)}$&        $\set{ 0 }$, $\set{ 1 }$, $\set{ 2 }$, $\set{ 3 }$, $\set{ 0, 1, 2 }$, $\set{ 0, 1, 3 }$, $\set{ 0, 2, 3 }$, $\set{ 1, 2, 3 }$, $\set{ 4 }$, $\set{ 5 }$, $\set{ 6 }$, $\set{ 7 }$, $\set{ 0, 1, 4 }$, $\set{ 0, 1, 5 }$, $\set{ 0, 1, 6 }$, $\set{ 0, 2, 4 }$, $\set{ 0, 3, 4 }$, $\set{ 8 }$, $\set{ 9 }$, $\set{ 10 }$, $\set{ 11 }$, $\set{ 12 }$\\\hline        
        
$\mathbf T^{(0)}$&$\set{}$,$\set{ 0, 1 }$, $\set{ 0, 2 }$, $\set{ 0, 3 }$, $\set{ 1, 2 }$, $\set{ 1, 3 }$, $\set{ 2, 3 }$, $\set{ 0, 1, 2, 3 }$, $\set{ 0, 4 }$, $\set{ 0, 5 }$, $\set{ 0, 6 }$, $\set{ 0, 7 }$, $\set{ 1, 4 }$, $\set{ 1, 5 }$, $\set{ 1, 6 }$, $\set{ 1, 7 }$, $\set{ 2, 4 }$, $\set{ 2, 5 }$, $\set{ 2, 6 }$, $\set{ 3, 4 }$, $\set{ 3, 5 }$, $\set{ 0, 1, 2, 4 }$, $\set{ 0, 8 }$, $\set{ 0, 9 }$
        $\set{ 0, 10 }$,
        $\set{ 0, 11 }$\\\hline                   \end{tabular}
\label{tPatterns}
\end{table}
We perform decoding of $(n,n-1,2)$ codes by testing a few pre-defined error patterns $\mathcal E^{(i)}$.  This method is known as Chase-II decoding algorithm \cite{chase72class} and was used in \cite{sarkis2016fast}. First, the codeword symbols are arranged in the increasing order of their reliabilities, so that  $|S_{t[0]}|\leq |S_{t[1]}|\leq \dots |S_{t[n-1]}|$.  Second, a hard decision vector $\hat c$  is constructed, and its parity $p$ is calculated. Then the codewords are constructed as $c^{(i)}=\hat c+e^{(i)}$, where $e^{(i)}$ is the vector containing 1's on positions $t[\epsilon_{i,j}]$ and 0's elsewhere, for all $\mathcal E^{(i)}=\set{\epsilon_{i,0},\dots,\epsilon_{i,w_i}}\in \mathbf T^{(p)}$. The set of test error patterns $\mathbf T^{(p)}$ can be constructed either analytically using the expressions derived in \cite{fossorier1995softdecision}, or by simulations. Table \ref{tPatterns} presents the test error patterns used in BSDA. These patterns were obtained via simulations. It turns out that the same set of test error patterns can be used for decoding of codes of arbitrary length without any noticeable performance loss compared with the optimal decoder. 

The Chase-II decoding may result in performance degradation of BSDA, since the considered algorithm does not necessarily return true $L$ most reliable codewords. In this case one should increase the size of $\mathbf T^{(p)}$ and/or reduce the maximal allowed length of single parity check (SPC) code. 
Observe that the fast implementation of SCL in \cite{sarkis2016fast} uses only 8 test error patterns. In our implementation we use 26 ones. Furthermore, simulations show, that in most cases it is sufficient to identify the positions $t[0],t[1]$ of only two least reliable symbols. This can be done using the tournament algorithm \cite{KnuthArt3}.

\subsection{Double parity check codes}
 A $(n,n-2,2)$ polar code with the set of frozen symbol indices $\mF=\set{0,1}$ can be obtained by interleaving two $(n/2,n/2-1,2)$ codes.  This enables one to decode such codes using a combination of two decoders of a SPC code.

\subsection{Rate-1 code}
\label{ssRate-1}
For $(n , n)$ codes we propose to use the same decoding algorithm as for SPC codes (see Section \ref{ssSPC}). Moreover, simulations show that finding just 4 (out of $2^n$) most probable codewords  of $(n,n)$ code  does not result in any noticeable performance loss for considered error rates. These 4 most probable codewords is obtained by considering the following
error patterns: $\emptyset,\set{0}, \set{1}, \set{0,1}, \set{2}$. Their computation requires identification only $3$ smallest values $|S_j|, j \in [n]$.

\subsection{$(16,10,4)$, $(16,11,4)$ and $(16,12,2)$ codes}
These codes, obtained by Plotkin concatenation of $(8,4,4)$ or $(8,3,4)$ codes and $(8,7,2)$ or $(8,8,1)$ codes, commonly arise in the PD of polar codes.  Decoding of these codes can be implemented using the approach introduced in \cite{ivanov2016hybrid}.

\section{Block sequential decoding for polar subcodes}
\label{sPolSub}
\subsection{Dynamic frozen symbols}
It was suggested in \cite{trifonov2016polar} to set frozen symbols $u_i, i\in \mF$ not to zero, but to linear combinations of some other symbols,  i.e. 
\begin{equation} 
\label{mDynFrozen}
 u_i=\sum_{s=0}^{i-1}V_{j_i,s}u_s,\end{equation} 
  where $V$ is a $(n-k)\times n$ binary matrix, such that its rows end in distinct columns, and $j_i$ is the index of row with the last non-zero element in column $i$. Such symbols with non-trivial right hand side expressions are called dynamic frozen symbols (DFS), and the code obtained via considered construction are referred to as \textit{polar subcodes}. Decoding of such codes can be implemented by a straightforward generalization of the successive cancellation algorithm and SC-based algorithms.

 Properly constructed polar subcodes may have higher minimum distance than classical polar codes. This results in substantially better performance  \cite{trifonov2016polar,trifonov2018randomized} under the SCL algorithm and other SC-based algorithms. Polar codes with CRC \cite{tal2015list} can be considered as a special case of polar subcodes. 

A system of dynamic freezing constraints (DFC) may be constructed for any  linear code of length $2^m$ with check matrix $H$, by setting $V=QHA_m^T$, where $Q$ is a suitable invertible matrix. This enables one to decode such code with  the SCL algorithm \cite{trifonov2016polar}.

\subsection{Processing of dynamic frozen symbols}
\label{sDynFrozenProcessing}
Decoding of polar subcodes requires one to compute the  values of DFS, i.e. some linear combinations of symbols $v_i$ for any path $v_0^{\phi_{\psi_l}}$. The Tal-Vardy list decoding algorithm does not store these values explicitly. It is possible to express their values from the content of arrays $C_{l,\lambda}$. However, we employ an alternative approach, which is more efficient in practice.

In most cases, polar subcodes have only a few non-trivial DFS which depend on a small number of other symbols.  Let $f$ be the number of non-trivial equations \eqref{mDynFrozen} for the considered code. It can be assumed without loss of generality that these equations correspond to $f$ topmost rows of matrix $V$. Let $i_s \in \mF, 0\leq s<f$, be the indices of the corresponding dynamic frozen symbols. 
Let $\mathcal P=\set{j\:|\:V_{s,j}=1,0\leq j<i_s,0\leq s<f}$ be the set of indices of symbols participating in any of the DFC.
  
We propose to allocate boolean variables $w_{l,s}, 0\leq l<D$ for each path, initialize them to 0 at decoder startup, and flip the value of $w_{l,s}$ at each phase $j<i_s$, such that $V_{s,j}=1$ and $v_{j}=1$, where $v_{j}$ is the value of the $j$-th symbol on the $l$-th path. Then at phase $i_s$ the value of $w_{l,s}$ is exactly the value of the $s$-th DFS for the corresponding path.
 However, the above described BSDA does not compute explicitly the values $v_{j}$. But one can obtain these values as $v_{j}=(\hat c A_{m_{\psi_l}})_{j\bmod 2^{m_{\psi_l}}}$, where $\hat c$ is a codeword of an outer code obtained for path $l$ at block $\psi_l$. This approach is illustrated in Figure \ref{fPrepareDFEvaluation}.
\begin{figure}
\small
\begin{algorithm}{PrepareForDFEvaluation}{l,\hat c}
\begin{FOR}{j\in \mathcal P\cap\set{j\:|\:\phi_{\psi_l}-2^{m_{\psi_l}}<j'\leq \phi_{\psi_l}}}
v_{j} \leftarrow (\hat c A_{m_{\psi_l}})_{j\bmod 2^{m_{\psi_l}}}\\
\begin{IF}{v_{j}=1}
\begin{FOR}{i\=0 \TO f-1}
w_{l,i}\=w_{l,i}+V_{i,j}
\end{FOR}
\end{IF}
\end{FOR}
\end{algorithm}
\caption{Accumulating the values of dynamic frozen symbols}
\label{fPrepareDFEvaluation}
\end{figure}
Observe that lines 2 and 4--5 of the algorithm can be efficiently implemented via bit mask manipulation techniques. 

If there is a non-trivial DFS in some block $\psi_l$, i.e. $\phi_{\psi_l}-2^{m_{\psi_l}}<i_{s}\leq \phi_{\psi_l}$ for some $s$, and $v_{l,s}=1$ when the decoder reaches this block, then one should perform decoding in a non-trivial coset of the corresponding outer code. The coset representative is given by 
\begin{equation}
\mathbf p_s=\sum_{\substack{i=\phi_{\psi_l}-2^{m_{\psi_l}}+1\\V_{s,i}=1}}^
{\phi_{\psi_l}}(A_{m_{\psi_l}})_{i\bmod 2^{m_{\psi_l}},-},
\end{equation}
where $B_{i,-}$ denotes the $i$-th row of matrix $B$. 

We introduce the algorithm \textit{GetCoset}, which computes the value
$$
\mathbf p=\sum_{s\in \mathfrak{S}}
w_{l,s}\mathbf p_s,
$$
where 
$
\mathfrak{S} =\set{s|\phi_{\psi_l}-2^{m_{\psi_l}}<i_{s}\leq \phi_{\psi_l}}.
$ 
During the decoding process, \textit{GetCoset} should be called before the line 3 of the \textit{ForwardPass} function (Figure \ref{fForwardPassBSDA}) and  adjust the signs of the LLRs $\mathbf S$, i.e. $\mathbf S_i \leftarrow \mathbf S_i (-1)^{\mathbf p_i}$. After that, each codeword $\hat c$ returned from the outer decoder (including \textit{BackwardPass} function) should be corrected according to $\mathbf p$, i.e. $\hat c \leftarrow \hat c + \mathbf p$. Note that the vectors $\mathbf p_s$ can be precomputed offline. The corrected vector $\hat c$ should be passed to \textit{PrepareForDFEvaluation} function to update values $w_{l,s}$.

\begin{figure*}[t]
%\centering
\begin{subfigure}{0.5\textwidth}
\includegraphics[width=\textwidth]{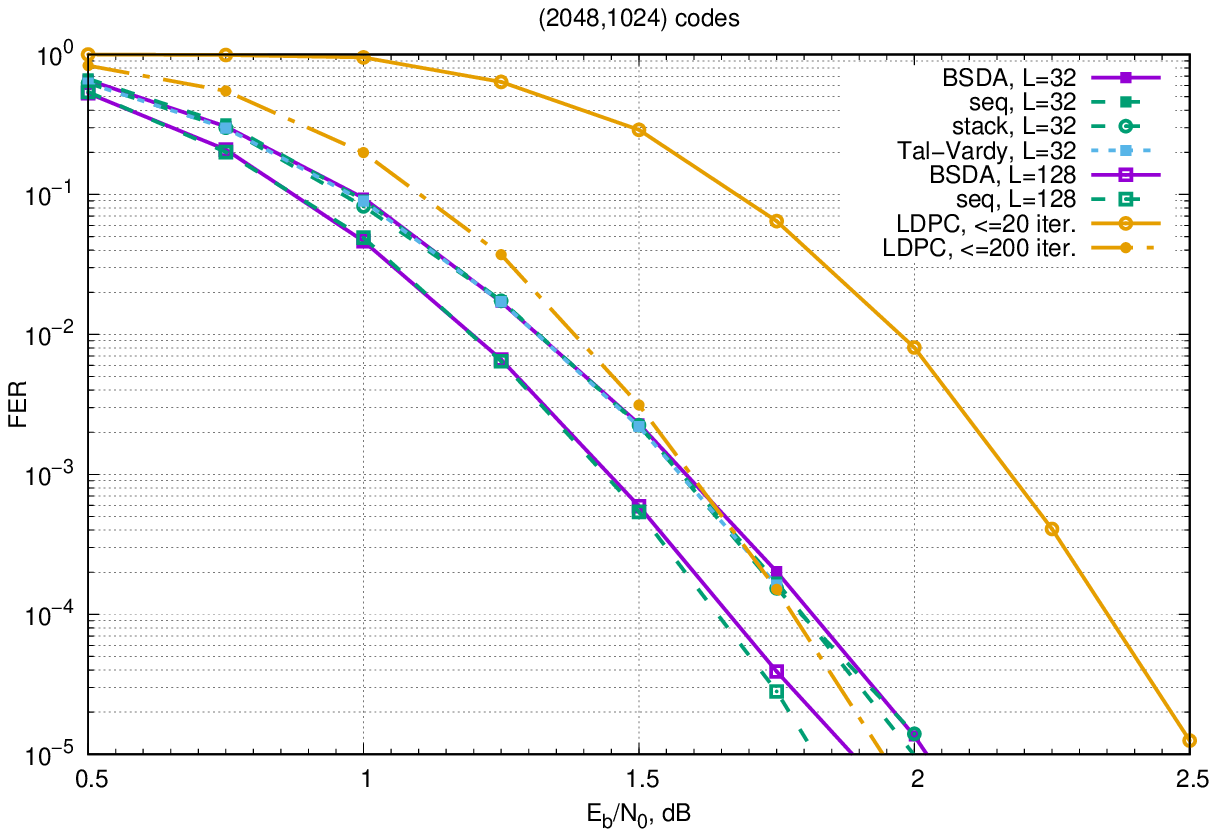}
\caption{Performance BSDA    }\label{fPerformance}
\end{subfigure}
\begin{subfigure}{0.5\textwidth}
\includegraphics[width=\textwidth]{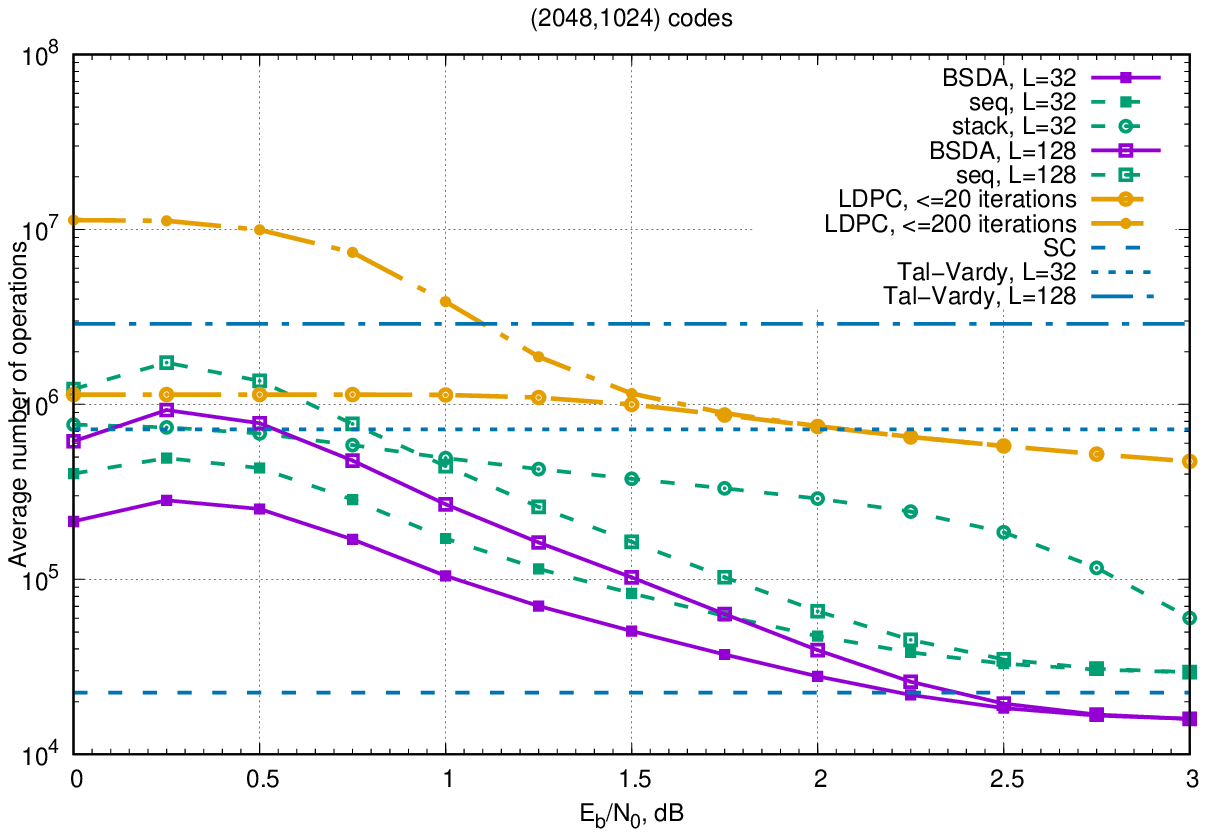}
\caption{Complexity  }\label{fAvgComplexity}
\end{subfigure}
\caption{Performance and complexity   of decoding algorithms for $(2048,1024)$ codes }
\end{figure*}

\section{Complexity analysis}
\label{sComplexity}
%\vspace{-2.5mm}
Consider  block sequential decoding of $(n=2^m,k)$ polar code. Let $\mathcal V$ be a number of outer codes. The worst-case complexity of the proposed decoding algorithm corresponds to the case when  exactly $L\mathcal V$ iterations are performed, i.e. $q_\psi=L, 0\leq \psi<\mathcal V$. In this case the number of operations performed by the decoder is given by 
\begin{equation}
\label{mComplexity}
C\leq L\sum_{\psi=0}^{\mathcal V-1}\left(C_\psi'+C_{\psi-1}''+
\Lambda(m-m_\psi,\floor{\phi_{\psi}/2^{m_{\psi}}})\right),
\end{equation}
where $C_\psi'$ is the complexity of a call to \textit{Preprocess} and \textit{GetNextCodeword} (see \textit{ForwardPass} function)  for outer code $\mathcal C_\psi$,  $C_\psi''$ is the complexity of subsequent calls\footnote{We assume $C_{-1}''=0$.} to \textit{GetNextCodeword} (see \textit{BackwardPass}).  Here  $\Lambda(\lambda,\phi)=2^{m-\lambda}(2^{d+1}-1)$ is the complexity of  computing  $S_{\lambda}^{(\phi)}$ via function \textit{IterativelyCalcS}, where $d=d(\phi)<\lambda$ is the maximal integer, such that $2^{d(\phi)}|\phi$.
\\\indent

Application of the proposed approach makes sense only if  \textit{Preprocess} and \textit{GetNextCodeword} functions provide a simpler way to obtain $L$ most probable codewords of $\mathcal C_\psi$ compared with the Tal-Vardy algorithm\footnote{$L$ must be sufficiently large to ensure that the Tal-Vardy algorithm always finds $L$ most probable codewords.} with list size $L$. Hence, the worst-case complexity of the proposed approach can be upper-bounded by considering the case (this corresponds to the algorithm presented in \cite{miloslavskaya2014sequential}) of $m_i=0$. In this case one has $C_\psi'=C_\psi''=0, 0\leq \psi<\mathcal V=n$, and 
$$
C\leq L\sum_{\phi=0}^{n-1}\Lambda(m,\phi)=\sum_{\phi=0}^{n-1}(2^{d(\phi)+1}-1).
$$

For any $d< m-1$  there are $2^{m-d-1}$ integers $\phi<2^m$ divisible by $2^d$ (and $2$ of them for $d=m-1$), but not divisible by $2^{d+1}$. Hence, one obtains 
$$
\label{mComplexityBound}
C\leq L(2^m-1+\sum_{d=0}^{m-1}2^{m-d-1}(2^{d+1}-1))=Lm2^m,
$$
which is identical to the complexity of the SCL decoding. The best case complexity corresponds to the case when the decoder visits each block 
%exactly
once, so it is given by $n \log_2 n$ with $L=1$.

%\indent\
There are additional costs associated with PQ operations. With appropriate implementation \cite{Cormen2001introduction,yakuba2015multilevel}, their complexity is upper bounded by $O(D\mathcal V)$.

\section{Numeric results}
\label{sNumRes}

Figure \ref{fPerformance} illustrates the performance of the proposed BSDA. Simulations were run for the case of AWGN\ channel, BPSK\ modulation and randomized polar subcode \cite{trifonov2017randomized}.
For comparison, we report also the performance of list \cite{tal2015list}, sequential \cite{trifonov2018score} and min-sum stack \cite{chen2013improved} decoding algorithms for the same code, and the CCSDS LDPC code under belief propagation decoding. 
It can be seen that the proposed algorithm provides essentially the same performance as the sequential and Tal-Vardy algorithms.  Furthermore, for $L=32$ its performance is close to that of the LDPC code with at most 200 decoder iterations. Even better performance  is obtained for $L=128$.  
\begin{figure}[h]
\centering
\includegraphics[width=0.49\textwidth]{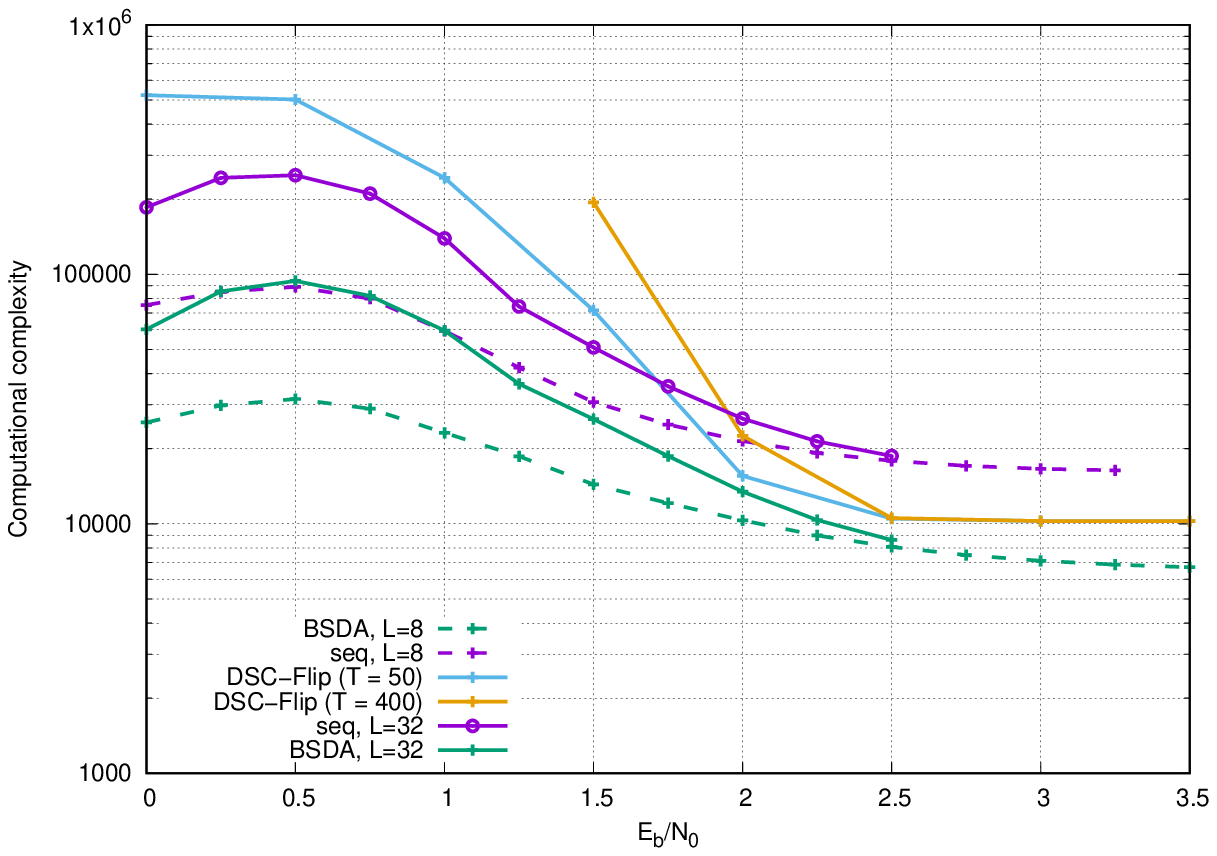}
\caption{Complexity of decoding algorithms, (1024,512) code}
\label{SCFlip}
\end{figure}

\begin{figure*}[h]
%\centering
\begin{subfigure}{0.5\textwidth}
\includegraphics[width=\textwidth]{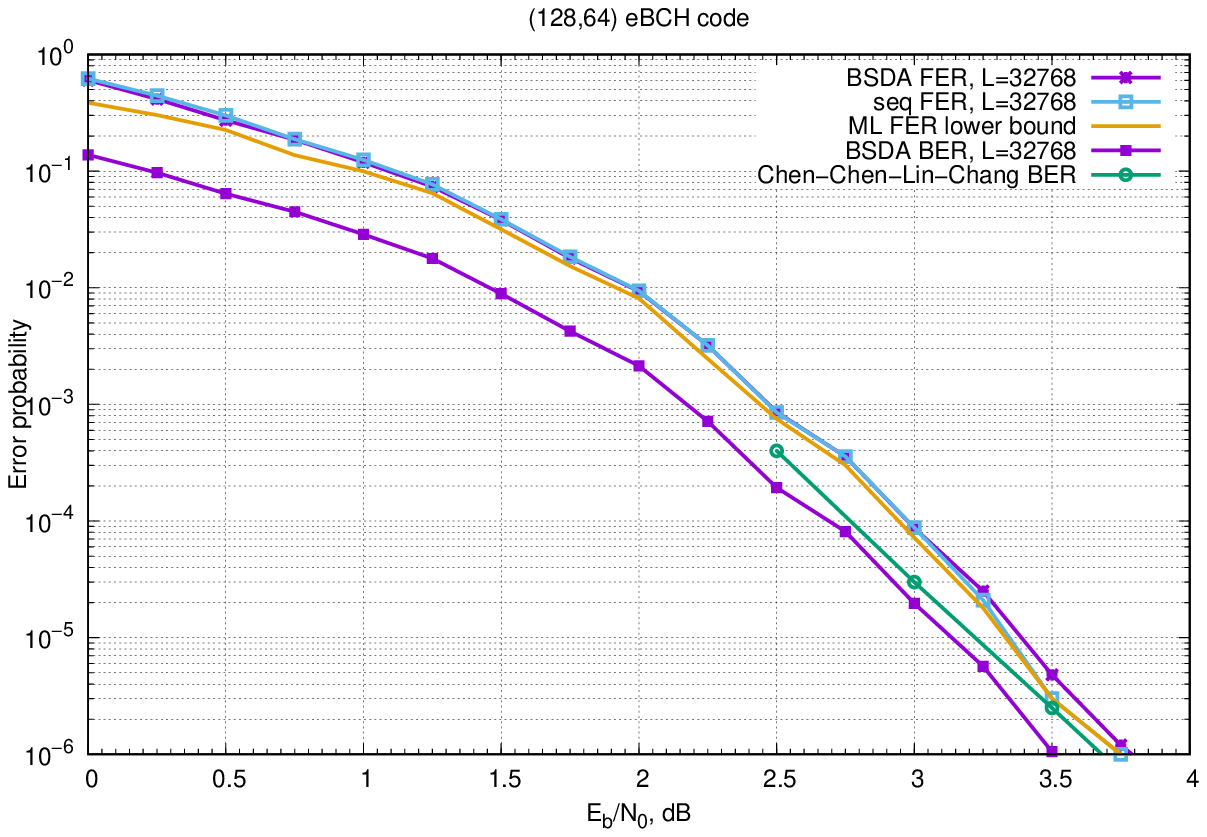}   
\caption{Performance}\label{fBCHPerformance}
\end{subfigure}
\begin{subfigure}{0.5\textwidth}
\includegraphics[width=\textwidth]{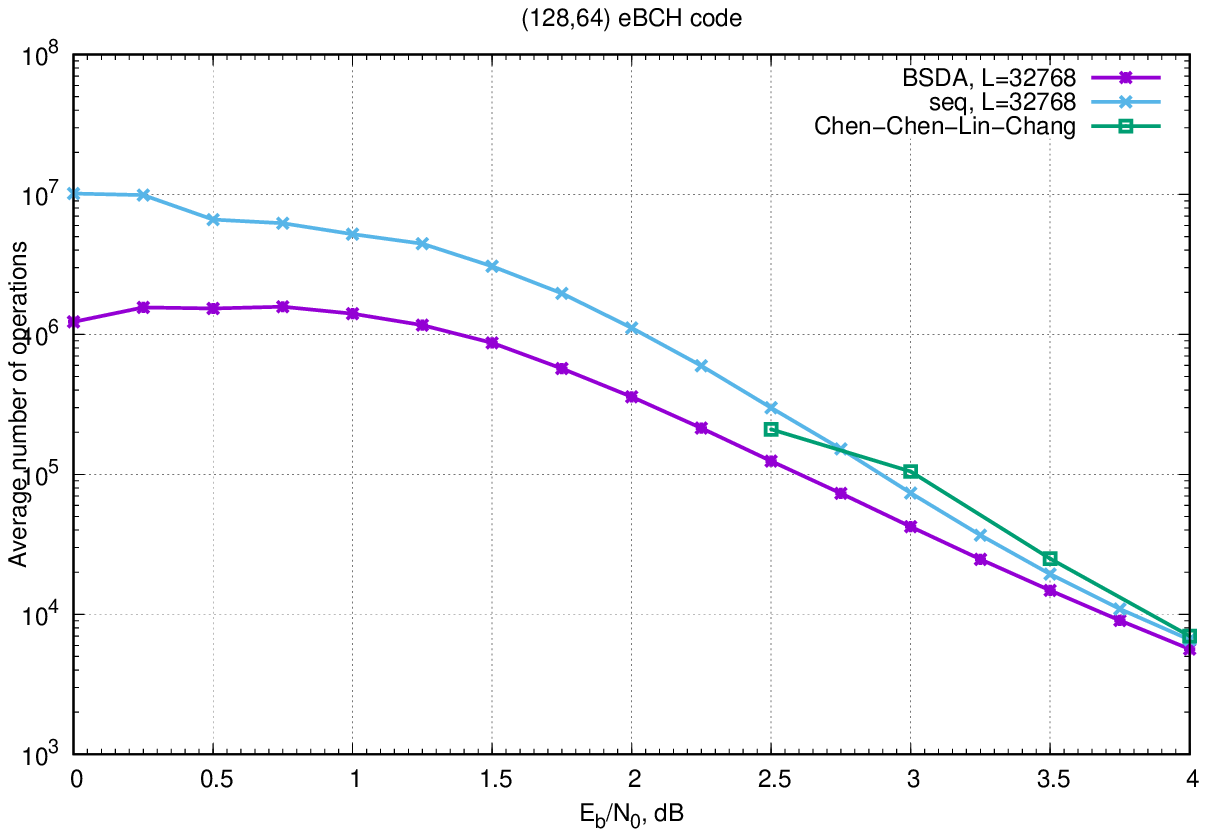}
\caption{Complexity}\label{fBCHComplexity}
\end{subfigure}
\caption{Block sequential  decoding    of $(128,64)$ eBCH  code }
\label{fBCH}
\end{figure*}

Figure \ref{fAvgComplexity} illustrates the average number of summation and comparison operations performed by the considered algorithms. 
It can be seen that the complexity of the SDA is much lower compared with the original stack algorithm (which corresponds to $\Psi(\phi)=0, 0\leq \phi<n$).
Furthermore, the average complexity of the block sequential algorithm  converges quickly to a value slightly less than $n\log_2 n$, the complexity of the SC algorithm. The complexity of the proposed algorithm is 1.5--2 times lower compared with that of the sequential decoder, and substantially lower compared with $Ln\log_2 n$, the complexity of the Tal-Vardy list decoding algorithm, and the average complexity of the min-sum stack decoding algorithm. It is also substantially lower compared with the complexity of the BP decoder for the LDPC code.
Observe that reducing the maximal number of iterations for the BP\ algorithm results in a noticeable performance degradation without significant complexity reduction for $FER<0.1$. 

Figure \ref{SCFlip} presents  the average complexity of BSDA, sequential and SC Flip  \cite{chandesris2018dynamicscflip} decoding algorithms. It can be seen that the complexity of the sequential decoder with $L=8$ becomes higher than D-SCFlip with parameter $T = 50$, which corresponds to $L=8$ in SCL.
On the contrary, BSDA has lower decoding complexity even with $L = 32$.

As mentioned in Section \ref{sPolSub}, any binary linear block code can be decoded with the proposed algorithm, although the performance of such a decoder depends strongly on the structure of the corresponding frozen set. eBCH codes were shown to have sufficiently low SC decoding error probability \cite{trifonov2016polar}, and are therefore well-suited for decoding using the BSDA. Figure \ref{fBCH}  illustrates performance and complexity of the decoding of $(128,64,22)$ eBCH code. For comparison, we report also the results for Chen-Chen-Lin-Chang algorithm (a sequential-type trellis-based decoding method), reproduced from  \cite{chen2015algorithms}. It can be seen that the BSDA provides lower decoding complexity.

\begin{figure*}[h]
%\centering
\begin{subfigure}{0.5\textwidth}
\includegraphics[width=\textwidth]{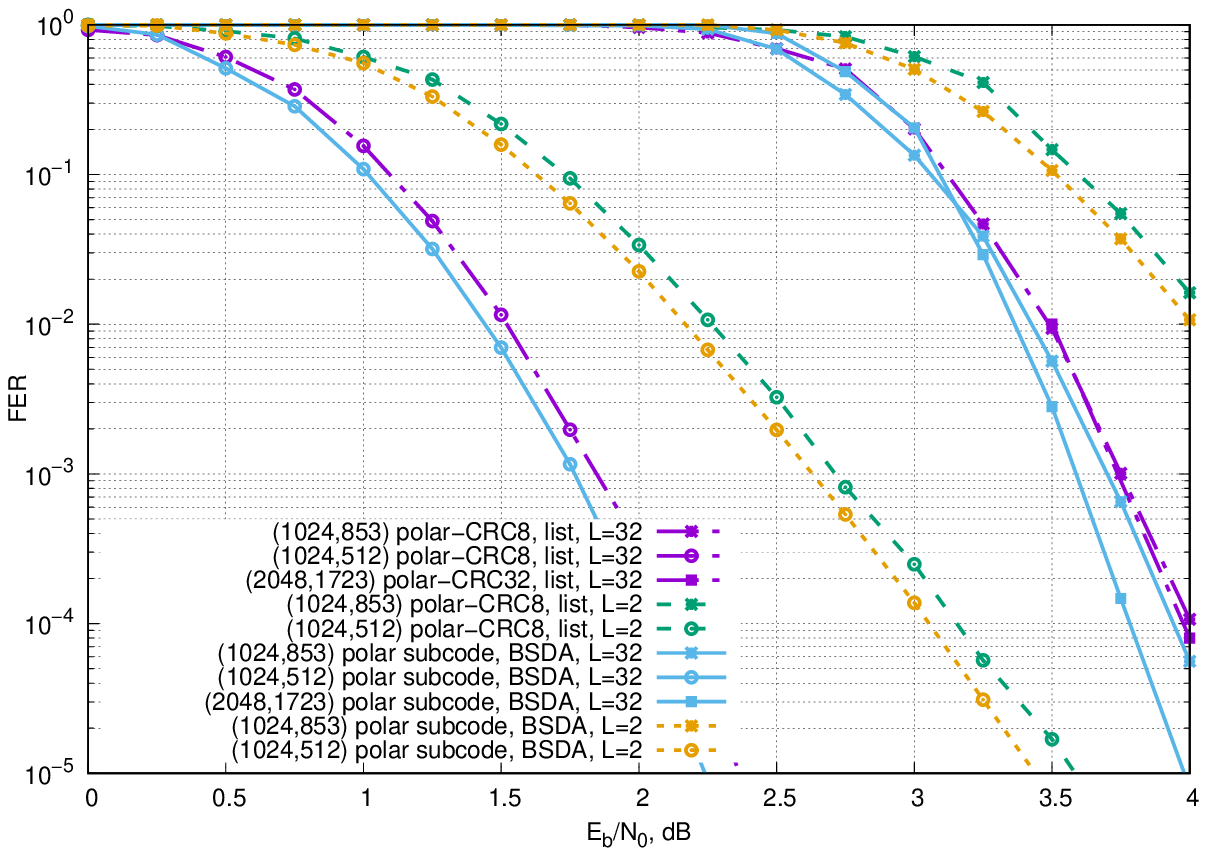}
\caption{Performance }\label{fSarkisPerformance}
\end{subfigure}
\begin{subfigure}{0.5\textwidth}
\includegraphics[width=\textwidth]{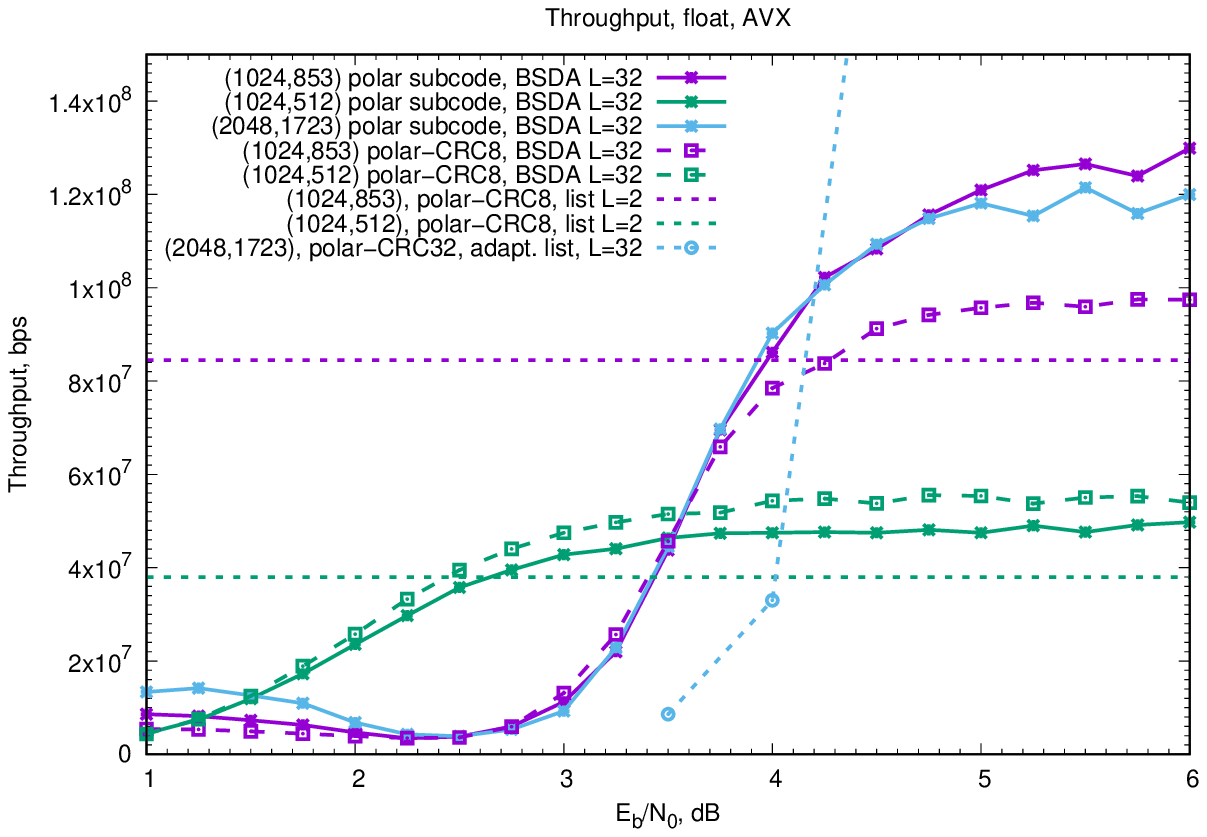}
\caption{Throughput of the software implementation }
\end{subfigure}
\caption{Performance and throughput of block sequential and fast list decoding algorithms  }
\label{fSarkisFight}
\end{figure*}
Figure \ref{fSarkisFight} illustrates the performance and throughput of the software implementation of the proposed BSDA, as well as fast list and adaptive list (ASCL) decoding algorithms introduced in \cite{sarkis2016fast}, 
for the case of polar subcodes and polar codes with CRC-8. 
 Simulations were performed on Intel Core i7-2600K CPU running at 3.4 GHz with maximum turbo frequency 3.8 GHz. SIMD techniques introduced in \cite{gal2015multigbs,sarkis2016fast}, based on single-precision floating point arithmetic, were used to implement LLR computation in the proposed algorithm. Throughput results for the fast and ASCL decoding algorithms are reproduced from \cite{sarkis2016fast}. The performance of polar codes with CRC under the BSDA is very close to that of the list decoder with the same $L$, and is therefore not shown. 
As it may be expected, polar subcodes provide better performance than polar codes with CRC, and increasing list size $L$ results in better performance. One can see that, for polar subcodes, at sufficiently high SNR the proposed BSDA even for $L=32$ provides the same or even better average throughput as the fast list decoding algorithm introduced in \cite{sarkis2016fast} for polar codes with CRC and $L=2$. Furthermore, at high SNR the throughput of BSDA for polar codes with CRC exceeds that  of the fast list decoding algorithm.   Observe that the algorithm presented in \cite{sarkis2016fast} relies on unrolling to eliminate redundant calculations, i.e. the decoder is specific for each code. The proposed BSDA is generic, but still provides higher throughput despite of much more sophisticated flow control structure. 

It can be also seen that for a $(2048,1723)$ polar subcode and $E_b/N_0<4.2$ dB the BSDA provides higher throughput and substantially better performance compared with the ASCL decoding algorithm \cite{sarkis2016fast} for a polar code with CRC-32. However, for higher values of $E_b/N_0$ the throughput of the ASCL decoding becomes much higher. The reason for this is that in this case with high probability the decoding is successful already with $L=1$ (i.e. with plain SC decoding), and this can be easily verified by CRC. Hence, highly complex list decoder is almost not used. It is, however, not clear how to extend the idea of adaptive list decoding to the case of polar subcodes, which provide much better performance.

Table \ref{tMemory}
presents  the amount of memory used by the decoder in various scenarios. The value of $\Xi$  is the maximal amount of memory (in Kilobytes) sufficient for storing arrays $S,C$, and outer decoder state variables $Z$ from the common memory pools, described in Appendix \ref{sMemory}. The values of parameters $L$, $D$ were selected to minimize overall memory demand during block sequential decoding, while ensuring that the performance does not degrade with respect to the case of $D=Lk$, which corresponds to the maximal possible memory footprint. Minimization for each code was carried out for FER at $10^{-3}$. Observe that for an SCL\ decoder one needs to store $Ln$ LLRs $S$ and  $2Ln$ partial sums $C$. For a software implementation, this results in $6Ln$ bytes of storage. The corresponding values are shown as $\Xi_{TV}$ in the table. It can be seen that $\Xi/\Xi_{TV}$ decreases with code length.

\begin{table}[]

        \caption{Decoder memory requirements }
        \label{tMemory}\footnotesize
        \centering
        \begin{tabular}{|l|l|l|l||l||l|l|l|}
                \hhline{----||----}
                {$L$} & {$D$}  & {$\Xi$} &{$\Xi_{TV}$} &  {$L$} & {$D$}  & {$\Xi$, KB}&{$\Xi_{TV}$} \\ \hhline{----||----}
                \multicolumn{4}{|c||}{ (1024, 512, 28)} & \multicolumn{4}{|c|}{ (16384, 8192, 48)} \\ \hhline{----||----}
                8      & 70                      &   385 & 49 & 8      & 250                     &   4764  & 786      \\ \hhline{----||----}
                32     & 240                    &   1457  &197& 32     & 900                    &   18617   &  3146   \\ \hhline{----||----}
                256    & 1620                   &   11254 &1572& 256    & 8230                   &   160791    &  25165 \\ \hhline{----||----}
                \multicolumn{4}{|c||}{ (2048, 1024, 48)} & \multicolumn{4}{|c|}{ (2048, 683, 52)}    \\ \hhline{----||----}
                8      & 100                    &   786    & 98   & 8      & 100                    &   681      & 98  \\ \hhline{----||----}
                32     & 370                    &   3071  &   393  & 32     & 400                   &   2686       & 393\\ \hhline{----||----}
                256    & 3020                   &   24215  &3146    & 256    & 2450                   &   20859     & 3146 \\ \hhline{----||----}
%                 \multicolumn{3}{|c||}{Polar subcode (4096, 2048, 48)}  & \multicolumn{3}{|c|}{Polar subcode (2048, 1440, 24)}  \\ \hhline{----||----}
%                 8      & 130                    &   1332        & 8      & 110                    &   813         \\ \hhline{----||----}
%                 32     & 700                   &   5523        & 32     & 330                   &   3160        \\ \hhline{----||----}
%                 256    & 4470                  &   43303       & 256    & 2460                   &   25334       \\ \hhline{----||----}
%                 \multicolumn{3}{|c||}{Polar subcode (8192, 4096, 48)}  & \multicolumn{3}{|c|}{Polar subcode (8192, 6553, 18)}  \\ \hhline{----||----}
%                 8      & 200                     &   2374        & 8      & 200                     &   2851        \\ \hhline{----||----}
%                 32     & 710                    &   9175        & 32     & 910                     &   11252       \\ \hhline{----||----}
%                 256    & 6250                   &   73493       & 256    & 5870                   &   91052       \\ \hhline{----||----}
        \end{tabular}
\end{table}

\section{Conclusions}
In this paper the block sequential decoding algorithm was introduced. It employs blockwise processing of the input symbols of the polarizing transformation. The processing
%It was shown that the input symbols of the polarizing transformation can %be processed blockwise, and the processing 
operation reduces to on-demand construction of codewords of the codes arising in the Plotkin decomposition of the code being decoded. A set of such codes was identified, which admits low complexity list decoding. 

It was shown that the proposed block sequential decoding algorithm has lower complexity than the sequential, stack and list decoding algorithms, while having approximately the same performance. At sufficiently high SNR, the throughput of the software implementation of the proposed algorithm exceeds the throughput of the fast list decoder with much smaller list size, i.e. the proposed algorithm provides better performance and lower decoding complexity compared with the list decoding algorithm by Sarkis et al \cite{sarkis2016fast}. The proposed algorithm can be used for decoding of polar (sub)codes, polar codes with CRC and short eBCH\ codes.

\appendix
\label{sImplementation}
\section{Low-Level Algorithms}
\label{sLowLevelAlg}
\subsection{Data structures and basic procedures}
The proposed decoding algorithm can be implemented using the techniques suggested in \cite{tal2015list}. However, several simplifications are possible.
Let $l,\lambda,\phi,\beta$ denote the path, layer, phase and branch number, respectively. Each path is associated with arrays of intermediate LLRs $S_{l,\lambda}[\beta], 0\leq l<D, 0\leq \lambda\leq m-\mu,0\leq \beta<2^{m-\lambda},$ where $D$ is the maximal number of paths considered by the decoder (i.e. the maximal size of the PQ), and $2^\mu, 0\leq \mu<m,$ is the length of the shortest outer code in the Plotkin decomposition tree, i.e. $\mu = \min_{\psi \in[\mathcal V]} m_\psi$. Each path is also associated with value $R_l$, which contains values $R(u_0^\phi|y_0^{n-1})$, similarly to \cite{balatsoukasstimming2015llrbased,balatsoukasstimming2014hardware}.

It was suggested in \cite{tal2015list} to store the arrays of partial sum tuples $C_{l,\lambda}[\beta][\phi\bmod 2]$. We propose to rename these arrays to $C_{l,\lambda,\phi\bmod 2}[\beta]$. By examining the \textit{RecursivelyUpdateC} algorithm presented in \cite{tal2015list}, one can see that $C_{l,\lambda,1}[\beta]$ is just copied to $C_{l,\lambda-1,\psi}[2\beta+1]$ for some $\psi\in\set{0,1}$, and this copy operation terminates on some layer $\lambda'$. Observe that  $\lambda-\lambda'$
is equal to the maximal integer  $d$, such that $\phi+1$ is divisible by $2^d$.
 Therefore, we propose to co-locate $C_{l,\lambda,1}[\beta]$  with $C_{l,\lambda_0,0}[\beta]$. In this case the corresponding pointers are given by $C_{l,\lambda,1}=C_{l,\lambda_0,0}+2^{m-\lambda}(2^{\lambda-\lambda'}-1)$. This not only results in the reduction of the amount of data stored by a factor of two, but also enables one to   avoid "copy on write"  operation (see line 6 of Algorithm 9 in \cite{tal2015list}). Therefore, we write $C_{l,\lambda_0}$ instead of $C_{l,\lambda_0,0}$ in what follows.

We use the array pointer mechanism suggested in \cite{tal2015list}  to avoid data copying. However, we distinguish the case of read and write data access.  Retrieving  read-only  pointers is performed by functions \textit{GetArrayPointerC\_R}$(l,\lambda)$ and \textit{GetArrayPointerS\_R}$(l,\lambda)$ shown in Figure \ref{fROAccess}. 
Retrieving writable pointers is performed by function \textit{GetArrayPointerW}$(T,l,\lambda)$, where $T\in \set{'C','S'}$ shown in Figure \ref{fWAccess}. This function implements reference counting mechanism similar to that proposed in \cite{tal2015list}.
It is discussed in more details in Section \ref{sMemory}. 

 Figures \ref{fIterativelyCalcS} and \ref{fIterativelyUpdateC} present
 iterative algorithms for computing $S_{l,\lambda}[\beta]$ and $C_{l,\lambda}[\beta]$.
 These algorithms resemble the recursive ones given in \cite{tal2015list}. However, the proposed implementation avoids costly array dereferencing operations. 

\begin{figure}
\small
\begin{subfigure}{0.5\textwidth}
\begin{algorithm}{IterativelyCalcS}{l,\lambda,\phi}
d\=\max\set{0\leq d'\leq \lambda-1|\phi \text{is divisible by $2^{d'}$}}\\
\lambda'\=\lambda-d\\
S'\=\CALL{GetArrayPointerS\_R}(l,\lambda'-1)\\
N\=2^{m-\lambda'}\\
\begin{IF}{\text{$\phi 2^{-d}$ is odd}}
\tilde C\=\CALL{GetArrayPointerC\_R}(l,\lambda')\\
S''\=\CALL{GetArrayPointerS\_W}(l,\lambda')\\
S''[\beta]\=\CALL{P}(\tilde C[\beta],S'[\beta],S'[\beta+N]), 0\leq \beta<N\\
S'\=S'';\lambda'\=\lambda'+1;N \leftarrow N/2
\end{IF}\\
\begin{WHILE}{\lambda'\leq \lambda}
S''\=\CALL{GetArrayPointerS\_W}(l,\lambda')\\
S''[\beta]\=\CALL{Q}(S'[\beta+N],S'[\beta]), 0\leq \beta<N\\
S'\=S'';\lambda'\=\lambda'+1;N\=N/2
\end{WHILE}
\end{algorithm}
\caption{Computing  $S_{\lambda}^{(\phi)}(u_0^{\phi-1},y_0^{N-1})$}
\label{fIterativelyCalcS}
\vspace{0.5cm}
\end{subfigure}
%\end{figure}
%\begin{figure}
\small
\begin{subfigure}{0.5\textwidth}
\begin{algorithm}{IterativelyUpdateC}{l,\lambda,\phi}
\delta\=\max\set{d|\phi+1 \text{is divisible by
$2^{d}$}}\\
\tilde C\=\CALL{GetArrayPointerC\_W}(l,\lambda-\delta,0)\\
N\=2^{m-\lambda};
\tilde C=\tilde C+N(2^\delta-2); C''\=\tilde C+N;\\
\lambda'\=\lambda-\delta\\
\begin{WHILE}{\lambda>\lambda'}
C'\=\CALL{GetArrayPointerC\_R}(l,\lambda)\\
\tilde C[\beta]\=C'[\beta]\oplus C''[\beta], 0\leq \beta<N\\
N\=2N;
C''\=\tilde C;
\tilde C\=\tilde C-N\\
\lambda\=\lambda-1
\end{WHILE}
\end{algorithm}
%\vspace{3cm}
\caption{Updating $C$ arrays}
\label{fIterativelyUpdateC}
\end{subfigure}
\caption{Computing LLRs and partial sums}
\end{figure}

\subsection{Memory management}
\label{sMemory}
\begin{figure}
\small
%\begin{tabular}{ll}
\parbox{0.5\textwidth}{
\begin{algorithm}{GetArrayPointerW}{T,l,\lambda}
p\=PathIndex2ArrayIndex[l,\lambda]\\
\begin{IF}{p=-1}
p\=\CALL{Allocate}(\lambda)
\ELSE
\begin{IF}{ArrayReferenceCount[p]>1}
ArrayReferenceCount[p]--\\
p\=\CALL{Allocate}(\lambda)
\end{IF} 
\end{IF}\\
\RETURN ArrayPointer[T][p]\ 
\end{algorithm}}%&\parbox{0.5\textwidth}
{
\begin{algorithm}{GetArrayPointerS\_W}{l,\lambda}
\RETURN \CALL{GetArrayPointerW}('S',l,\lambda)
\end{algorithm}
\begin{algorithm}{GetArrayPointerC\_W}{l,\lambda,\phi}
\delta\=\max\set{d|\phi+1 \text{is divisible by
$2^{d}$}}\\
C\=\CALL{GetArrayPointerW}('C',l,\lambda-\delta)\\
\begin{IF}{\phi\equiv 1\bmod 2}
C\=C+2^{m-\lambda}(2^{\delta}-1)
\end{IF}\\
\RETURN C
\end{algorithm}}
%\end{tabular}
\caption{Write access to the data}
\label{fWAccess}
\end{figure}

\begin{figure}
%\begin{tabular}{ll}
%\scalebox{0.85}{\parbox{0.5\textwidth}
%{
\small
\begin{algorithm}{GetArrayPointerS\_R}{l,\lambda}
\RETURN\ ArrayPointer['S'][PathIndex2ArrayIndex[l,\lambda]]
\end{algorithm}
%&\scalebox{0.85}{\parbox{0.5\textwidth}
{
\begin{algorithm}{GetArrayPointerC\_R}{l,\lambda}
\RETURN\ ArrayPointer['C'][PathIndex2ArrayIndex[l,\lambda]]
\end{algorithm}
}%}
%\end{tabular}
\caption{Read-only access to the data}
\label{fROAccess}
\end{figure}

Many paths considered by the proposed algorithm share common values of $S_{l,\lambda}[\beta]$ and $C_{l,\lambda}[\beta]$, similarly to SCL decoding. To avoid duplicate calculations
one can use the same shared memory data structures. That is, for each path $l$ and for each layer $\lambda$ we store the index of the array containing the corresponding values $S_{l,\lambda}[\beta]$ and $C_{l,\lambda}[\beta]$. This index is given by $p=$\textit{PathIndex2ArrayIndex}$[l,\lambda]$, so that the corresponding data can be accessed as \textit{ArrayPointer}$[T][p], T\in\set{'S','C'}$.
Furthermore, for each integer $p$ we maintain the number of references to this array \textit{ArrayReferenceCount}$[p]$. If the decoder needs to write the data into an array, which is referenced by more than one path, a new array needs to be allocated. Observe that there is no need to copy anything into this array, since it will be immediately overwritten. This is an important advantage with respect to the implementation described in \cite{tal2015list}. However, the sequence of array read/write and stack push/pop  operations still satisfies the validity assumptions introduced in \cite{tal2015list}, so the proposed algorithm can be shown to be well-defined by exactly the same reasoning as the original SCL.

\begin{figure}
\small
\begin{algorithm}{Allocate}{\lambda}
[t,q]\=Pop(InactiveArrayIndices[\lambda]);\\
t\=t(m+1)+\lambda; 
ArrayReferenceCount[t]=1\\
\begin{IF}{q=1}
\begin{IF}{\Phi>\Lambda}
\CALL{Abort}
\end{IF}\\
ArrayPointer['S'][t]=PoolS+\Phi;\\ 
ArrayPointer['C'][t]=PoolC+\Phi\\
\Phi+=2^{m-\lambda}
\end{IF}\\
\RETURN t
\end{algorithm}
\caption{Adaptive memory allocation}
\label{fAllocate}
\end{figure}
Only one path considered by the decoder is constructed to the full length $n$. For most of the paths, only a few symbols are constructed before these paths are abandoned, i.e. either stored without being accessed till the decoder terminates, or killed. Hence, one does not need to provide the memory needed to accommodate all $D$ paths.
Therefore, we propose to create common memory pools for arrays $C$ and $S$, denoted \textit{PoolC} and \textit{PoolS}, respectively. If a new array needs to be provisioned, a part of memory pool is assigned to it. Arrays $C$ and $S$ are provisioned simultaneously.  Let $\Phi$ denote the amount of memory consumed from these pools.  If $\Phi$ exceeds the size of the memory pools $\Lambda$, then decoding needs to be terminated. For sufficiently large $\Lambda$ this typically occurs after the correct path has been killed by the decoder. 

The stack \textit{InactiveArrayIndices} stores pairs $[t, q]$, where $t$ is an index of an array, and q is true if the array is not allocated yet.
If the number of references to some array drops to 0, then the index of the array is saved in a stack of unused arrays, similarly to \cite{tal2015list}, so that it can be re-used later.   The indices of unused arrays corresponding to different layers are stored in different stacks \textit{InactiveArrayIndices}$[\lambda]$, since these arrays have different sizes. Allocation from the common pools occurs only if the corresponding index is extracted for the first time ($q=1$ in  Figure \ref{fAllocate}, which illustrates  the proposed approach).

\bibliographystyle{IEEETran}


\begin{IEEEbiography}[{\includegraphics[width=\textwidth]{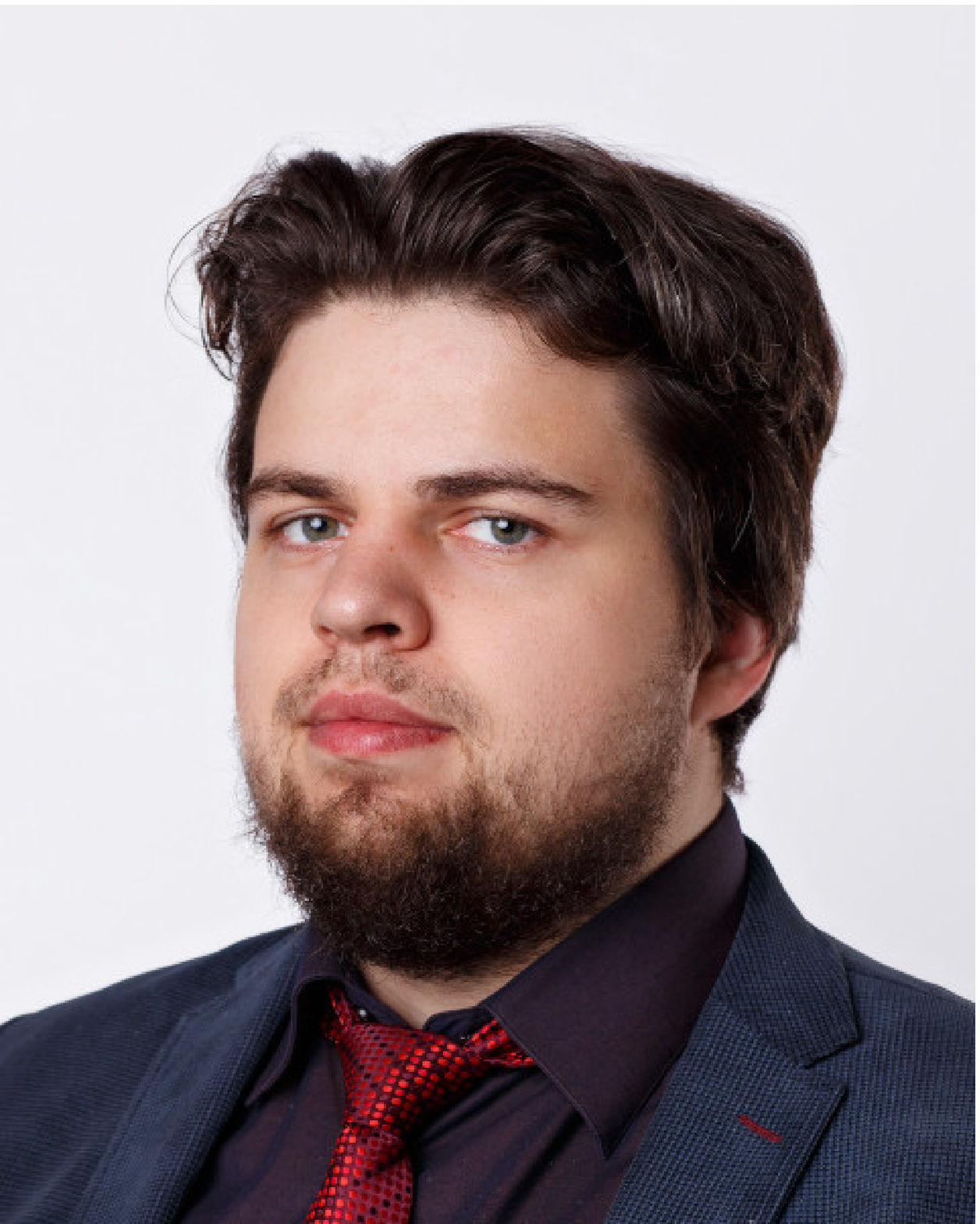}}]
{Grigorii Trofimiuk} (S'15)
was born in Boksitogorsk, Russia in 1994. He received 
the B.Sc. and M.Sc. degrees from St.Petersburg  Polytechnic
University in 2016 and 2018, respectively, all in computer science. He is currently working toward
the Ph.D. degree at the ITMO University in St.Petersburg, Russia. His research interests include coding theory and its applications in telecommunications. 

\end{IEEEbiography}
\begin{IEEEbiography}[{\includegraphics[width=\textwidth]{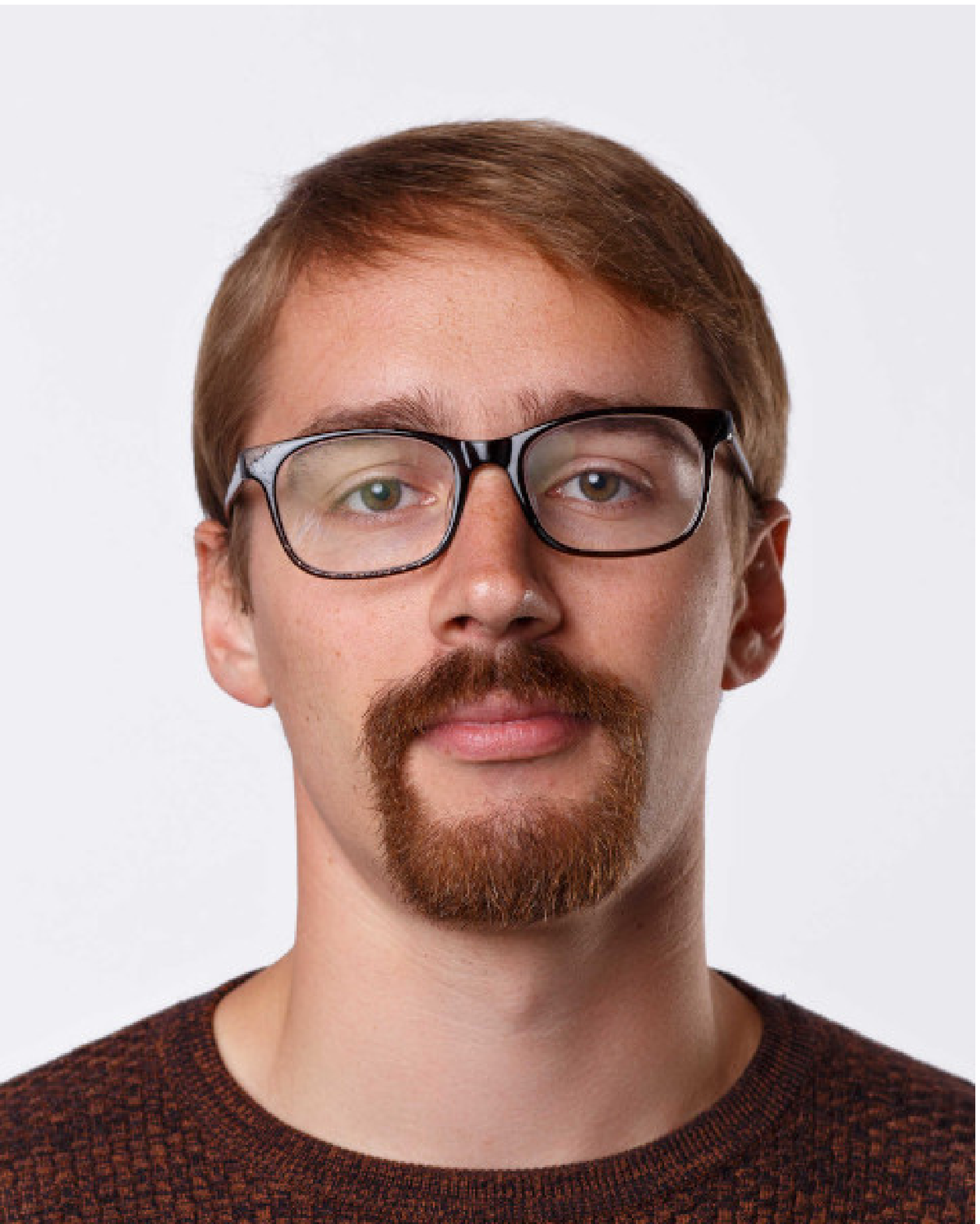}}]{Nikolai Iakuba} (S'15) was born in St.Petersburg and obtained his M.Sc. degree in St.Petersburg Polytechnic University. He is currently working toward
the Ph.D. degree at the ITMO University in Saint Petersburg, Russia. His research interests include coding theory, especially polar and Reed-Muller codes
\end{IEEEbiography}
\begin{IEEEbiography}[{\includegraphics[width=\textwidth]{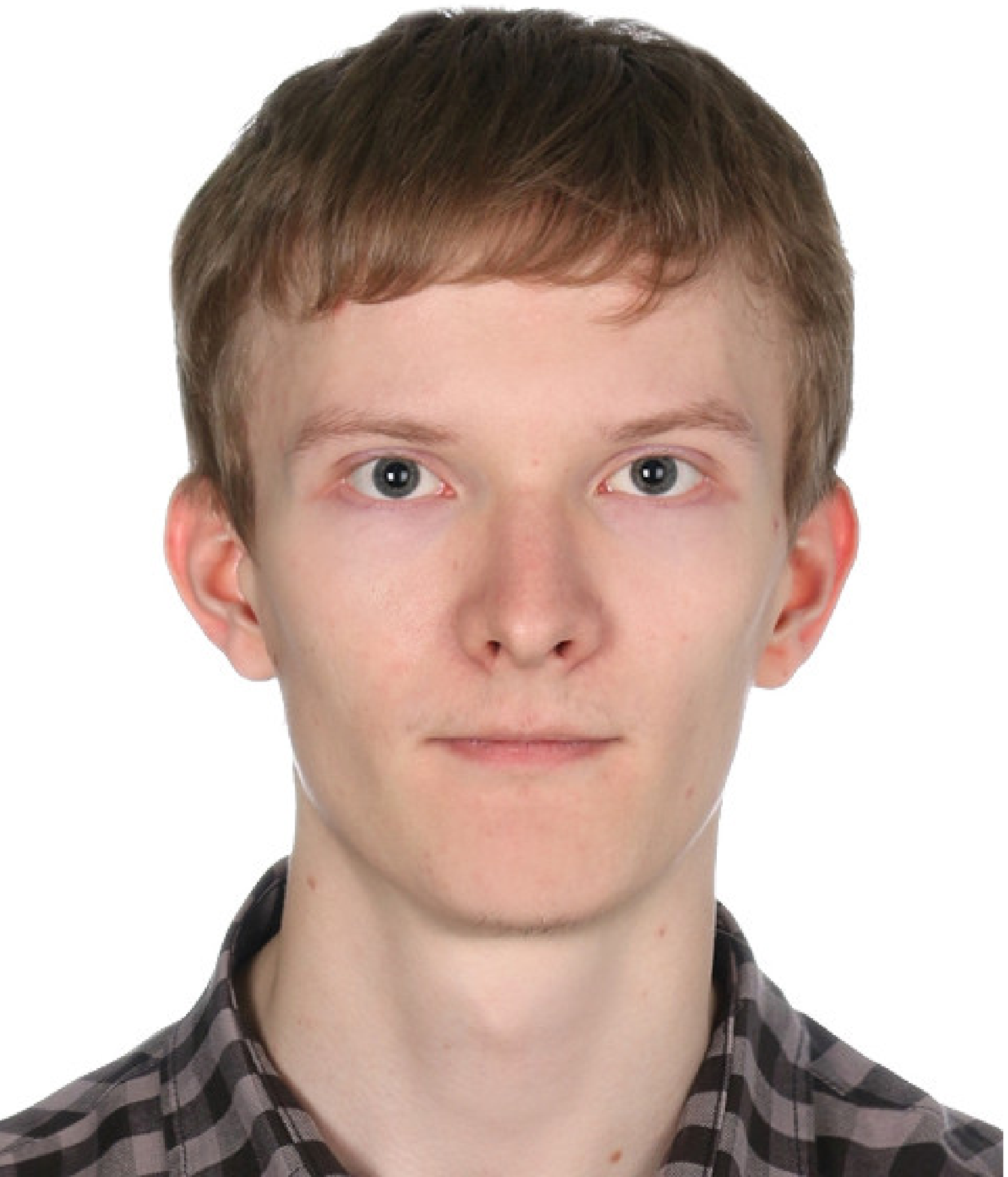}}]{Stanislav Rets} received the B.Sc. and M.Sc. degrees from St. Petersburg Polytechnic University, Russia, in 2015 and 2017, respectively. His current research interests include coded modulation techniques based on polar codes.
\end{IEEEbiography}
\begin{IEEEbiography}[{\includegraphics[width=\textwidth]{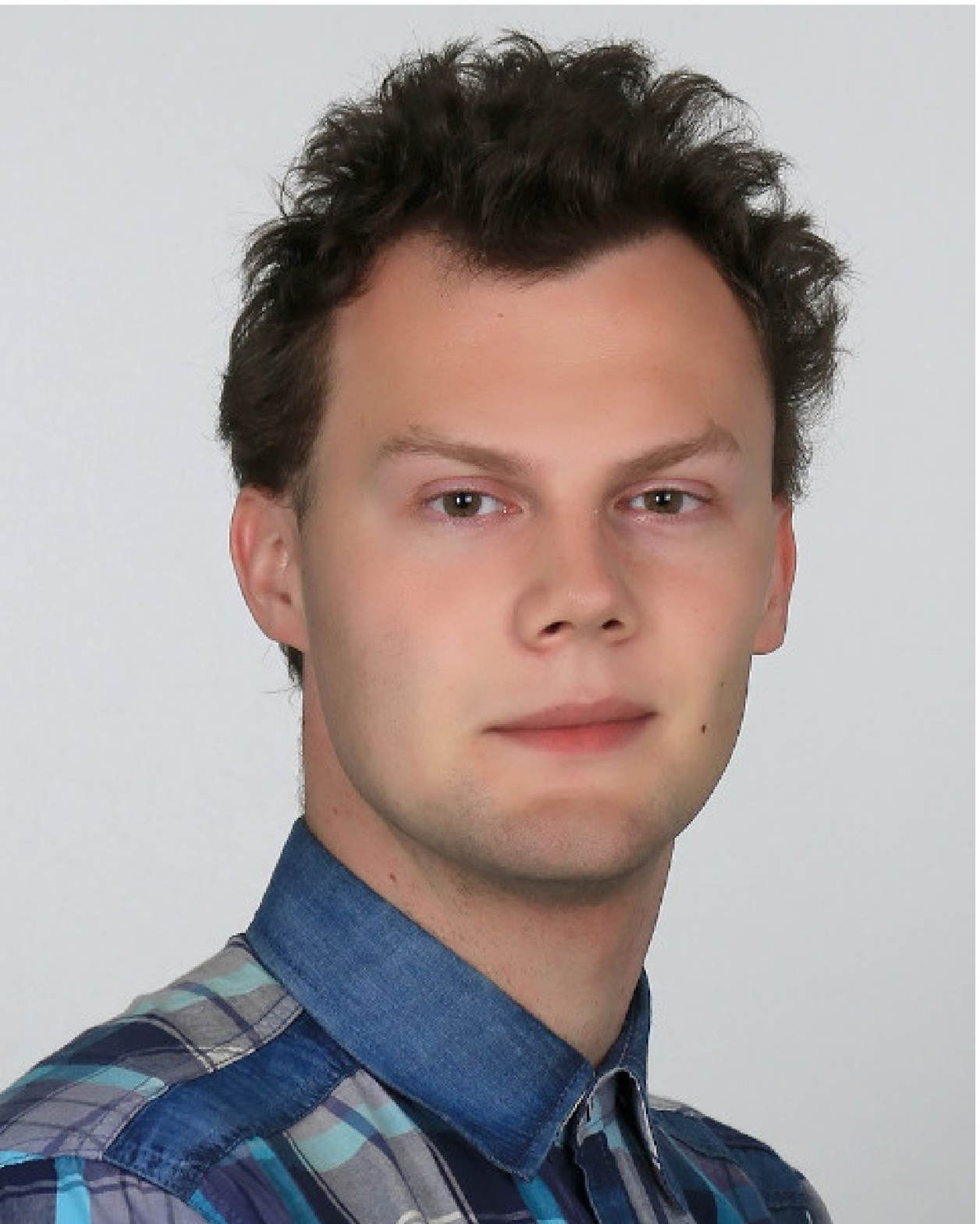}}]{Kirill Ivanov}(S'16) obtained his B.Sc. and M.Sc. degrees from St. Petersburg Polytechnic University, Russia, in 2015 and 2017, respectively. Currently he is a PhD student at \'Ecole polytechnique f\'ed\'erale de Lausanne, Switzerland under the supervision of Prof. R\"udiger Urbanke. 

His research interests include wireless communications systems and coding theory, with focus on polar and Reed-Muller codes. 
\end{IEEEbiography}
\begin{IEEEbiography}[{\includegraphics[width=\textwidth]{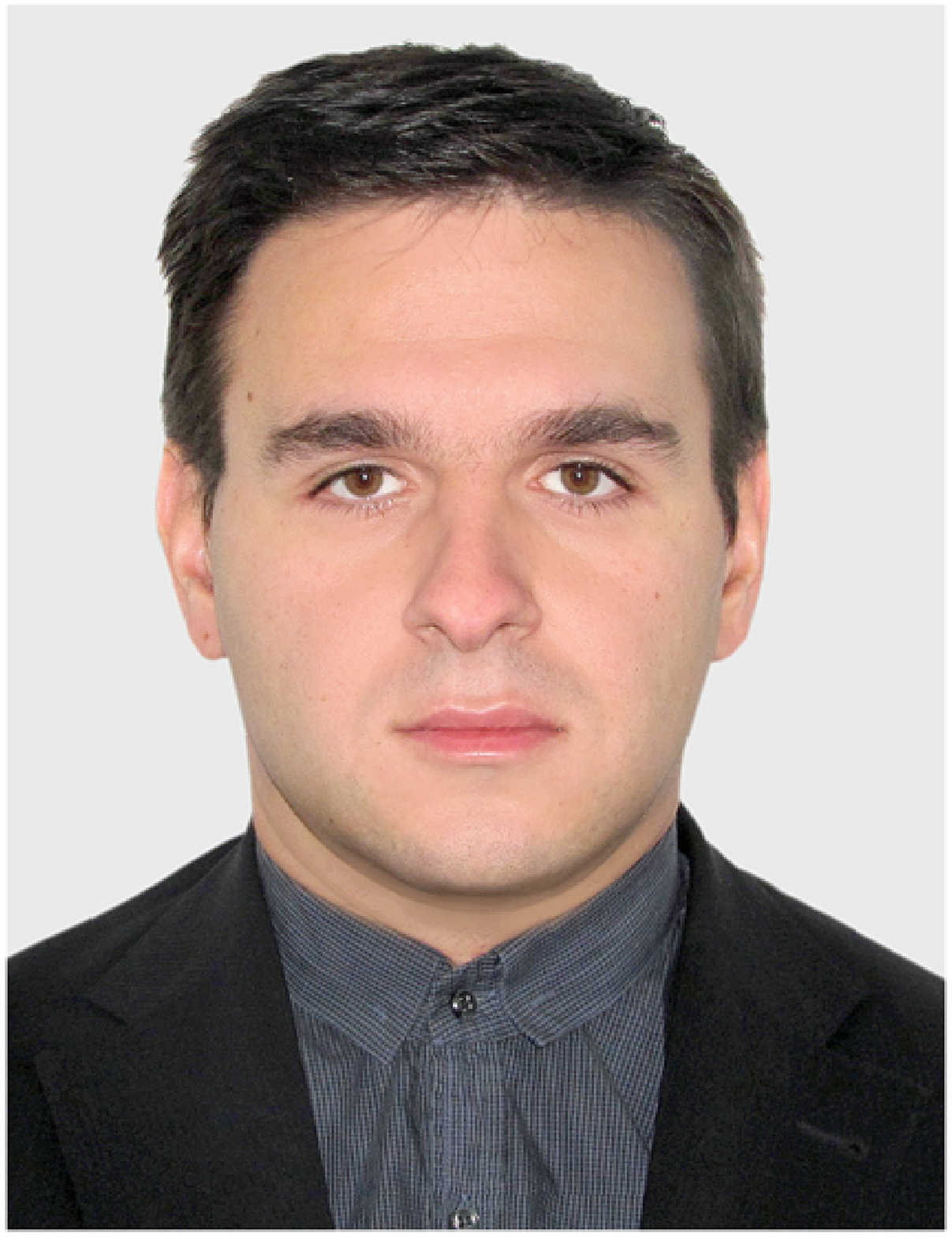}}]{Peter Trifonov} (S'02,M'05)
was born in St.Petersburg, USSR in 1980. He received 
the MSc and PhD (Candidate of Science) degrees from Saint Petersburg  Polytechnic
University in 2003 and 2005, and Dr.Sc degree from the Institute for Information Transmission Problems in 2018.  His research interests include coding theory and its applications in telecommunications and storage systems.  Currently he is a professor at the ITMO University in Saint Petersburg, Russia.
He is an editor at IEEE Transaction on Communications.
\end{IEEEbiography}

\end{document}